\documentclass[11pt]{article}
\usepackage[activeacute,english]{babel} 
\usepackage[utf8]{inputenc}
\usepackage{verbatim} 
\usepackage{graphicx} 
\usepackage[usenames,dvipsnames]{color} 
\usepackage{enumerate,mdwlist}  
\usepackage{multirow} 
\usepackage{amsmath,amsfonts,amssymb,amsthm} 
\usepackage{empheq} 
\usepackage{bbm,dsfont} 
\usepackage{mathrsfs} 
\numberwithin{equation}{section} 
\usepackage{tikz-cd,tkz-graph,tkz-arith,tkz-berge} 
\usepackage{tkz-graph,tkz-arith,tkz-berge} 
\usetikzlibrary{patterns,calc}
\usepackage[small,nohug,heads=littlevee]{diagrams}
\diagramstyle[labelstyle=\scriptstyle]
\usepackage{float} 
\usepackage[labelfont=bf,labelsep=space]{caption} 

\usepackage{geometry}
\makeatletter                   
\def\mdseries@tt{m}             
\makeatother                    
\usepackage[plain]{fancyref}
\usepackage[draft=true]{minted} 
\usepackage{color}
\usepackage{hyperref}           
\hypersetup{
    colorlinks=true,
    linkcolor=blue,
    filecolor=red,      
    urlcolor=magenta,
    breaklinks=true,            
}
\usepackage{breakurl}           

\usepackage{color}
\definecolor{red}{rgb}{.7,0,0}
\definecolor{blue}{rgb}{0,0,1}

\def\mcG{\mathcal{G}}
\def\mcH{\mathcal{H}}

\def\mcM{\mathcal{M}}

\def\mcP{\mathcal{P}}

\def\mcS{\mathcal{S}}
\def\mcO{\mathcal{O}}
\def\mcX{\mathcal{X}}
\def\mcU{\mathcal{U}}

\def\mcZ{\mathcal{Z}}

\def\mcW{\mathcal{W}}

\def\msC{\mathscr{C}}

\def\msO{\mathscr{O}}

\def\msK{\mathscr{K}}

\def\bbR{\mathbb{R}}
\def\bbZ{\mathbb{Z}}
\def\bbN{\mathbb{N}}

\def\bbS{\mathbb{S}}
\def\fkC{\mathfrak{C}}
\def\fkA{\mathfrak{A}}
\def\sgn{\mathsf{sgn}}
\def\sfK{\mathsf{K}}
\def\Tg{\mathrm{T}}
\usepackage{ifpdf}
\ifpdf
\usepackage[pdfencoding=auto,unicode]{hyperref} 
\hypersetup{
 pdftoolbar=true,
 pdfmenubar=true,
 pdfstartview={FitV},
 pdftitle={Computing the Homology of Semialgebraic Sets. I: Lax Formulas},
 pdfauthor={Peter B\"{u}rgisser, Felipe Cucker, Josu\'{e} Tonelli-Cueto},
 pdfsubject={semialgebraic geometry, algebraic topology, numerical algorithms},
 pdfkeywords={semialgebraic geometry,algebraic topology},
 pdfcreator={overleaf}, 
 pdfproducer={overleaf},
 linkcolor=red     
 citecolor=green   
 urlcolor=cyan     
 filecolor=magenta 
}
\fi
\usepackage{bookmark}

\title{Computing the Homology of Semialgebraic Sets.\\
I: Lax Formulas\thanks{This work was supported by the Einstein Foundation, Berlin.}}
\author{Peter B\"urgisser
\thanks{Partially funded by DFG research grant BU 1371/2-2.}
\\
Technische Universit\"at Berlin\\ 
Institut f\"ur Mathematik\\ 
GERMANY\\
{\tt pbuerg@math.tu-berlin.de} 
\and
Felipe Cucker
\thanks{Partially supported by a GRF grant
from the Research Grants Council of the Hong Kong
SAR (project number CityU 11302418).}
\\
Dept. of Mathematics\\
City University of Hong Kong\\
HONG KONG\\
{\tt macucker@cityu.edu.hk}
\and
Josu\'{e} Tonelli-Cueto\\
Technische Universit\"at Berlin\\ 
Institut f\"ur Mathematik\\ 
GERMANY\\
{\tt  ton-cue@math.tu-berlin.de}
}
\date{}

\makeatletter
\def\th@plain{%
  \thm@notefont{}
  \slshape 
}
\def\th@definition{%
  \thm@notefont{}
  \normalfont 
}
\makeatother
\theoremstyle{plain}
\newtheorem{lem}{Lemma}[section]
\newtheorem{prop}[lem]{Proposition}
\newtheorem{theo}[lem]{Theorem}
\newtheorem{cor}[lem]{Corollary}

\theoremstyle{definition}
\newtheorem{defi}[lem]{Definition}
\theoremstyle{remark}

\newtheorem{exam}[lem]{Example}
\newtheorem{obs}{Observation}[lem]
\newtheorem{remark}[lem]{Remark}
\newcommand{\eproof}{\hfill\qed}
\newcommand{\conv}{\mathsf{conv}}

\newcommand{\Ap}{\mathsf{S}}
\newcommand{\cost}{\mathop{\mathsf{cost}}}
\newcommand{\size}{\mathop{\mathsf{size}}}
\newcommand{\depth}{\mathop{\mathsf{depth}}}
\def\bfd{\boldsymbol{d}}
\def\Hd{\mcH_{\bfd}}
\def\Pd{\mcP_{\bfd}}
\def\hm{^{\mathsf{h}}}
\def\Hm{\mathsf{H}}
\def\kappabar{\overline{\kappa}}
\def\kappaff{\overline{\kappa}_{\sf aff}}

\def\codim{\mathrm{codim}}
\def\prob{\mathop{\mathrm{Prob}}}
\def\sfH{\mathop{\mathsf H}}
\newcommand{\cech}[2]{\text{\v{C}}_{#1}\big(#2\big)}
\def\Oh{\mathcal{O}}
\def\diff{\mathrm{D}}

\usepackage{multicol}
\usepackage[ruled,algosection,vlined]{algorithm2e}
\SetKwInOut{postcondition}{Postcondition}
\SetKwInOut{precondition}{Precondition}

\usepackage{rotating,lipsum}

\usepackage{afterpage}

\begin{document}
\maketitle
\begin{abstract}
We describe and analyze an algorithm for computing the homology (Betti numbers and torsion coefficients) of closed semialgebraic sets given by Boolean formulas without negations over lax polynomial inequalities. 
The algorithm works in weak exponential time. This means that 
outside a subset of data having exponentially small measure, the 
cost of the algorithm is single exponential in the size of the data.

All previous algorithms solving this problem have doubly 
exponential complexity (and this is so for almost all input data). 
Our algorithm thus represents an exponential acceleration 
over state-of-the-art algorithms for all input data outside a set 
that vanishes exponentially fast. 
\end{abstract}

\begin{quote}
    {\bf AMS classification numbers:} 14P10, 65D18, 65Y20, 68Q25
    
    
\end{quote}
\section{Introduction}

A {\em semialgebraic set} is a subset of $\bbR^n$ defined by a {\em Boolean combination} 
(i.e., a sequence of unions, intersections, and complements) of polynomial
equalities and inequalities.  The class of such sets is closed under
unions, intersections, complements, and projections as well as under
taking images and preimages of polynomial maps. This wealth of closure
properties is consistent with the wealth of shapes that semialgebraic
sets can take.

Semialgebraic sets play a distinguished role in several branches of
mathematics. In mathematical logic, they occur as the definable sets
of the first-order theory of real closed fields~\cite{Tarski51}, in
real algebraic geometry, where they are the constructible
sets~\cite{BR:90,BCR:98}, in complexity theory, where complete
problems over the reals in various complexity classes are stated in
terms of semialgebraic sets~\cite{BSS89,BC:09}, in mathematical
programming~\cite{BGH:05}, robotics~\cite{CannyThesis,SchSha:88},
\dots. Not surprisingly, in the last decades, a substantial amount of
work was devoted to the design of algorithms for problems involving
semialgebraic sets. This goal was already present in Tarski's
work~\cite{Tarski51}, where a procedure to decide the first-order
theory of the reals is given. 
In the 1970's Collins~\cite{Collins75} and W\"uthrich~\cite{Wut76},
independently introduced an algorithm, today known as {\em Cylindrical
Algebraic Decomposition} (and usually referred as CAD) that allowed to
solve most of the problems mentioned before.  The cost of running CAD
on a list of $q$ polynomials of degree at most $D$ in $n$ variables
is $(q D)^{2^{\mcO(n)}}$, that is, it has a doubly exponential
dependence on the number of variables, and this dependence is generic:
it does hold for all choices of coefficients for the polynomials in
the list outside a smaller dimensional set. In addition, since in all
applications one needs to first compute a CAD and only then solve from
this CAD the problem at hand, this generic doubly exponential
complexity appears to be unavoidable no matter the problem
considered. A new set of ideas, known as the {\em critical points
method}, was proposed in the late 1980's by Grigori'ev and
Vorobjov~\cite{GriVo88,Gri88}. Using these ideas, complexity bounds
improved to singly exponential in $n$, $(q D)^{\mcO(n)}$, for many
of the questions considered on semialgebraic sets: deciding
emptiness~\cite{GriVo88,Ren92a,BaPoRo96}, counting connected
components~\cite{BaPoRo99,canny:93,canny-grig-voro:92,GriVo92,HeRoSo:94},
computing the dimension~\cite{Koi98}, the Euler-Poincar\'e
characteristic~\cite{Basu_1996}, and the first few~\cite{Basu_2006}
Betti numbers.

A problem that was conspicuously left out of these improvements is
that of computing the sequence of homology groups of a semialgebraic
set. Of course, the list above contains partial results in this
direction (the number of connected components being the 0th Betti
number) but, as of today, no single exponential algorithm has been
devised returning the whole sequence of homology groups over $\bbZ$
(which gives information about the torsion in addition to the Betti
numbers). This sequence being arguably the most important set of
topological invariants, the importance of its computation can hardly
be overemphasized.
\smallskip

All the algorithms mentioned above are {\em symbolic} in the sense
that they assume infinite precision in the computations. If
implemented with finite precision they may, and experience shows 
they often do, return meaningless outputs. Driven by a search of
numerical stability, Cucker and Smale devised a {\em numerical}
algorithm for deciding emptiness~\cite{CS98}. The possibility of
round-off errors implies the existence of a set of inputs for
which, no matter the machine precision nor the algorithm at hand,
there exists computations that return a wrong answer (wrong number
of connected components, wrong dimension, etc.). This set of 
inputs, referred to as {\em ill-posed}, is usually lower 
dimensional and hence of measure zero in the space of data. 
Numerical algorithms are not expected to return an output on 
them. Instead, the computation time is expected to increase with
the proximity to the set of ill-posed inputs, a proximity which is
usually measured by, or closely related to, the {\em condition number}
of the input at hand. Complexity estimates are therefore commonly
expressed in the dimensions of the data and the condition number.

Yet, condition numbers are not, in general, known a priori, a fact that 
makes condition-based complexity bounds to be less informative. The 
standard way to overcome this drawback, pioneered by Goldstine and von
Neumann~\cite{vNGo51}, Demmel~\cite{Demmel88}, and
Smale~\cite{Smale97} among others, consists of accepting statistical
bounds in exchange of eliminating the condition number from these 
bounds.  To do so, the space of data is endowed with a probability
measure (usually a Gaussian on an Euclidean space or the uniform
distribution on a sphere) that allows one to treat the condition
number as a random variable. In this setting, the most common form of
analysis is the {\em average analysis}, that aims to bound the
expected running time of the algorithm in terms of the data's
dimensions only. Recently, however, Amelunxen and Lotz~\cite{AmLo:17}
brought in a different form of analysis with the aim of giving a
theoretical explanation of the efficiency in practice of numerical
algorithms whose average complexity was too high. A paradigm of this
situation is the power method to compute the leading eigenpair of a
Hermitian matrix: this algorithm is very fast in practice, yet the
average number of iterations it performs has been shown to be
infinite~\cite{Kostlan88}. Amelunxen and Lotz realized that here, as in 
many other problems, this disagreement between theory and practice is due to the
presence of a vanishingly small (more precisely, having a measure exponentially 
small with respect to the input size) set of outliers, outside of which the
algorithm can be shown to be efficient. Complexity estimates holding 
outside a set of exponentially small measure were called {\em weak} 
in~\cite{AmLo:17}. 

The lines above are the background against which we can state our main result.
\medskip

\noindent{\bf Statement of the main result.}\quad  
Let $n\geq 2$, $q\geq 1$, and $\bfd=(d_1,\ldots,d_q)$, with $d_i\geq 1$ 
for $i=1,\ldots,q$. We denote by $\Pd[q]$ the vector space of polynomial 
tuples $p=(p_1,\ldots,p_q)$ with $p_i\in\bbR[X_1,\ldots,X_n]$ of degree at most 
$d_i$. We let $D:=\max\{d_1,\ldots,d_q\}$ and denote by $N$ the dimension of 
$\Pd[q]$. 
We endow this space with the Weyl inner product (see~\S\ref{sec:Weyl} 
below) and the resulting unit sphere 
$\bbS(\Pd[q])=\bbS^{N-1}$ with the uniform probability measure. 

We say that a Boolean combination $\Phi$ over $p\in\Pd[q]$ is 
{\em lax} if it has no complements; only unions and intersections 
of the atomic sets $\{p_i(x)\le 0\}$, $\{p_i(x)=0\}$, 
and $\{p_i(x)\ge 0\}$, for $i=1,\ldots,q$ 
(see \S\ref{sec:logic}). 
The {\em size} of $\Phi$, denoted 
$\size(\Phi)$ is the number of unions and intersections in the sequence. 
We note that $\Phi$ defines a closed 
semialgebraic subset of $\bbR^n$ which we will denote by $W(p,\Phi)$. 
Finally, we associate to $p$ a condition number $\kappaff(p)$ (whose 
precise definition is in~\S\ref{sec:affine}). This condition number does not depend 
on $\Phi$. 

The {\em size} of a pair $(p,\Phi)$ is 
$\size(p,\Phi):=N+\size(\Phi)$. 

\begin{theo}\label{thm:main_result}
We exhibit a stable numerical algorithm~{\sf Homology} 
that, given a tuple $p\in\Pd[q]$ and a lax Boolean 
combination $\Phi$ over $p$, computes the homology groups of~$W(p,\Phi)$. 
The cost of~{\sf Homology} on input~$(p,\Phi)$,
denoted~$\cost(p,\Phi)$, that is, the number of arithmetic 
operations and comparisons in~$\bbR$, satisfies:
\begin{description}
\item[(i)] $\cost(p,\Phi)\leq 
\size(\Phi) q^{\Oh(n)} (nD\kappaff(p))^{\Oh(n^2)}$.
\end{description}
Furthermore, if $p$ is drawn from the uniform distribution on $\bbS^{N-1}$, then:
\begin{description}
\item[(ii)] $\cost(p,\Phi)\leq \size(\Phi) q^{\Oh(n)}
(nD)^{\Oh(n^3)}$ with probability at least 
$1-(nq D)^{-n}$, and 
\item[(iii)]
$\cost(p,\Phi)\leq 
2^{\mcO\big(\size(p,\Phi)^{1+\frac{2}{D}}\big)}$ 
with probability at least $1-2^{-\size(p,\Phi)}$.
\end{description}
\end{theo}

\begin{obs}
A few comments on~Theorem~\ref{thm:main_result} are
called for.
\begin{description}
\item[(i)]
A detailed explanation, along with a proof, of the numerical stability
mentioned in the statement above is in Section~7 of~\cite{CKS16}. As
the same explanation, word by word, applies in our context we will not
deal with this issue in the rest of our exposition.
\item[(ii)]
Part~(iii) of Theorem~\ref{thm:main_result} shows, in the terminology
introduced by~\cite{AmLo:17}, that {\sf Homology} works in {\em weak
exponential time}. 
\item[(iii)]
It is easy to check that all the routines in algorithm {\sf Homology} do parallelize.
The parallel version of the algorithm can then be shown to work in parallel time 
$\size(p,\Phi)^{\mcO(1)}$ with probability at least 
$1-2^{-\size(p,\Phi)}$. 
That is, it works in 
{\em weak parallel polynomial time}. We will be more precise 
in~\S\ref{sec:parallel}.
\end{description}
\end{obs}
\medskip

\noindent
{\bf Relations with previous work and new ideas.}\quad Our results have not grown 
in the vacuum. They owe ideas and intuitions to a number of works in
the literature. Our use of grids goes back to~\cite{CS98}.  The goal
in that paper was deciding feasibility of semialgebraic systems.
Subsequently, these ideas were extended to the problem of counting the
solutions of 0-dimensional real projective
sets~\cite{CKMW1,CKMW2,CKMW3} and, much more recently, to the
computation of homology groups. In~\cite{CKS16}, it is the homology of
real projective sets, and in~\cite{BCL17}, that of basic semialgebraic
sets. Some of the objects in these two papers ---notably the 
{\em algebraic neighborhoods} of a semialgebraic set 
given by $f=0$, $g\ge 0$ that were introduced in the last one 
(see~\S\ref{se:alg-neighborhoods} below)--- play a crucial 
role in our development. Yet, the road-map to compute the homology groups of the 
closed set $W$ at hand passes, in both papers, through computing a  
covering $\mcU$ of $W$ by open balls of the same radius such that the nearest-point 
map $\mcU\to W$ induces a deformation retraction. When $W$ is a general 
closed semialgebraic set,
however, such a covering may well not exist as the nearest point may be undefined 
at points arbitrarily close to $W$. A simple example of such a set $W$ 
is given in the following picture. 

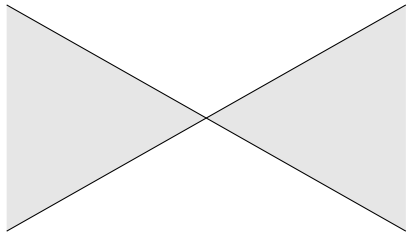
\begin{figure}[h]
\begin{center}
\begin{tikzpicture}[scale=0.75,point/.style={draw,minimum size=0pt,
    inner sep=1pt,circle,fill=black}]
\fill[black!10!white] (-3.5,-2) -- (-3.5,2) -- (0,0);
\fill[black!10!white] (3.5,-2) -- (3.5,2) -- (0,0);
\draw (-3.5,-2) -- (3.5,2);
\draw (-3.5,2) -- (3.5,-2);
\end{tikzpicture}
\end{center}
\end{figure}
We therefore need to proceed differently (see Remark~\ref{rem:wfs} below for more on 
this need). Our idea is to aim for a covering $\mcU$ 
which is only {\sl homologically} equivalent to $S$. To obtain it, we decompose $W$ 
as the union of sets $S_i$ given in terms of intersections only (basic semialgebraic sets)
and repeatedly use Mayer-Vietoris sequences to recover the homology of $W$ from the 
homology of these pieces. This requires to consider a family of 
algebraic neighborhoods for each $S_i$ and to use homological algebra to establish 
isomorphisms between the homology groups of the $S_i$, those of their algebraic 
neighborhoods, and those of their coverings $\mcU_i$. One key ingredient to make this 
possible is the Semialgebraic Triangulation Theorem. Another key ingredient is 
ensuring that all the algebraic neighborhoods above are 
homotopically equivalent to their corresponding $S_i$ and probably the major technical 
effort in our agenda is to estimate a size (or tolerance) for these algebraic 
neighborhoods that guarantees this equivalence. 
We do this (in Section~\ref{se:geometry}) in terms of the condition number. Our estimate 
quantifies the results of Durfee~\cite{Durfee1983}. Its proof relies on an explicit 
construction of Whitney stratifications and submersions for which Thom's First 
Isotopy Lemma~\cite{thom:69}  
can be applied. This use of semialgebraic geometry, as well as of 
differential and algebraic topology, sets our arguments 
apart from those in~\cite{CKS16} and~\cite{BCL17}. 
\medskip

\noindent
{\bf Future work.}\quad We are currently working in two directions extending this 
paper. On the one hand, to design an algorithm that works for arbitrary (i.e., not 
necessarily closed and given by lax formulas) semialgebraic sets. On the other hand, 
to show that the exponential weak complexity holds for a class of probability 
distributions more general than the ones we consider here.
\medskip
\goodbreak

{\small
\tableofcontents
}
\medskip

\noindent
{\bf Acknowledgments.}\quad We are grateful to Saugata Basu, Pierre Lairez 
and Nicolai Vorobjov for helpful discussions. 
\section{Overview of the Algorithm}\label{sec:overview}

In this section we describe, with broad strokes, the various steps and ingredients 
of algorithm {\sf Homology}. 

\subsection{Boolean combinations and propositional logic}\label{sec:logic}

There is a close relationship between lax Boolean combinations of equalities and 
inequalities and propositional logic. Indeed, any such Boolean combination 
over $p\in\Pd[q]$ corresponds to a propositional formula~$\Phi$ over $3q$ 
propositional variables 
$w^{\propto_j}_j$, $\propto_j\in\{\le,=,\ge\}$, $j\in\{1,\ldots,q\}$, built using 
the connectives $\vee$ and $\wedge$. 
Our Boolean combination is obtained by replacing $w^{\propto_j}_j$ by 
$\{p_j(x)\propto_j 0\}$ as well as $\vee$ by $\cup$ and $\wedge$ by $\cap$. We 
will freely use this correspondence all along this paper. In particular, we 
will use the expression {\em lax formula over $p\in\Pd[q]$} to denote a propositional 
formula as described above. 

A class of formulas of particular importance is that of {\em purely conjunctive} 
(usually referred to as {\em clauses} in the context of mathematical logic). Such  
formulas have the form $\bigwedge v_i$ where $v_i$ is either a variable or the 
negation of a variable (in the case of lax formulas only the first case is 
possible). Subsets of $\bbR^n$ defined by formulas of this kind over tuples 
$p\in\Pd[q]$ are called {\em basic semialgebraic}. 

A formula $\Psi$ is said to be in {\em Disjunctive Normal Form} (or in DNF for 
short) when it has the form 
$$
   \Psi\equiv \bigvee_{i\in I} \psi_i
$$
with $\psi_i$ purely conjunctive. It is a well-known fact 
(see, e.g.,~\cite[Theorem~3 in Section~2.3]{Hodel}) that for 
every propositional formula $\Phi$ there exists an 
equivalent formula $\Psi$ in DNF in the same set of variables. 
The same holds true when restricted to lax formulas. 
In our context this implies that, for all $p\in\Pd[q]$, 
the sets $W(p,\Phi)$ and $W(p,\Psi)$ coincide. 

\subsection{Homogenization}\label{sec:homog}

As before, let $\bfd =(d_1,\ldots,d_q)$ be a $q$-tuple of positive integers. 
We denote by $\Hd[q]$ the 
vector space of $q$-tuples $f=(f_1,\ldots,f_q)$ of  homogeneous polynomials 
$f_i\in\bbR[X_0,X_1,\ldots,X_n]$ of degree~$d_i$. 
We put $\bfd^* := (1,\bfd)$.
The {\em homogenization map} $\Hm:\Pd[q]\rightarrow 
\mcH_{\bfd^*}[q+1]$ is defined by
\begin{equation*}
  p\mapsto \Hm(p):=(\|p\|X_0,p_1\hm,\ldots,p_q\hm) ,
\end{equation*}
where $p_i\hm:=p_i\left(X_1/X_0,\ldots,X_n/X_0\right)X_0^{d_i}$ denotes 
the homogenization of $p_i$ and $\|p\|$ stands for the Weyl norm of 
the tuple~$p$ defined in \S\ref{sec:Weyl}. 

Any formula $\Phi$ over $f\in\Hd[q]$ 
defines a semialgebraic subset $\Ap(f,\Phi)$ of the sphere~$\bbS^n$. 
It will be convenient to call these sets {\em spherical semialgebraic}. 
In order to simplify the notation, we will also write $\Ap(f=0)$ etc.\ with 
the obvious meaning. 

The following result relates, for $p\in\Pd[q]$ and a formula $\Phi$
over $p$, the topology of the semialgebraic subset $W(p,\Phi)$ of the 
Euclidean space~$\bbR^n$ 
with that of the intersection of the spherical semialgebraic subset $\Ap(\Hm(p),\Phi)$
with the halfspace $X_0\ge 0$. We note that such 
a result would be straightforward if one intersected with $X_0>0$ instead.

\begin{prop}\label{generaltospherical}
Let $p\in\Pd[q]$ such that $\kappaff(p)<\infty$, let $\Phi$ be a lax formula over~$p$ 
and denote by $\Phi\hm$ the formula over $\Hm(p)\in\mcH_{\bfd^*}[q+1]$ given by
\[
  \Phi\hm:=\Phi(p_1\hm,\ldots,p_q\hm)\wedge \big(\|p\|X_0\ge 0 \big).
\]
Then the spaces $W(p,\Phi)$ and $\Ap(\Hm(p),\Phi\hm))$ are homotopically equivalent.
\end{prop}


We will prove Proposition~\ref{generaltospherical}  in~\S\ref{sec:affine}. 
It allows us to assume, in all 
that follows, that we are dealing with spherical semialgebraic sets 
given by Boolean combinations of homogeneous polynomials. We will 
freely use in this new context the terminology introduced 
in~\S\ref{sec:logic}. 

\subsection{Estimation of the condition number}

In \S\ref{sec:kappabar} we will associate a condition number
$\kappabar(f)\in [1,\infty]$ to a tuple 
$f\in\Hd[q]$ whose inverse measures how near are the intersections between the 
hypersurfaces given by $f$ from being non-transversal (the condition
number $\kappaff(p)$ 
in Theorem~\ref{thm:main_result} is actually $\kappabar(\Hm(p))$).  
This condition number provides information on the geometry of every possible 
spherical semialgebraic set built from $f$. Tuples $f$ for which 
$\kappabar(f)=\infty$ are said to be {\em ill-posed}. They are precisely those 
tuples for which there exists a formula $\Phi$ such that arbitrary 
small perturbations of $f$ may change the topology of $\Ap(f,\Phi)$. 
The set $\overline{\Sigma}_{\bfd}[q]$ of ill-posed tuples has positive codimension in 
$\Hd[q]$ and $\kappabar(f)$ estimates how far is $f$ from $\overline{\Sigma}_{\bfd}[q]$. 

The first substantial computational effort performed by {\sf Homology} 
is to estimate the condition number $\kappabar(f)$ of a tuple 
$f\in\Hd[q]$. The following result, which we will prove in~\S\ref{sec:est-kappa}, 
deals with this task.

\begin{prop}\label{prop:kappa-est}
There is an algorithm {\sc $\kappabar$-Estimate}
that, given $f\in\Hd[q]$, returns a number~$\sfK$ such that 
$$
   0.99\,\kappabar(f) \ \leq\ \sfK \ \leq\ \kappabar(f) 
$$
if $\kappabar(f)<\infty$,  
or loops forever otherwise. The cost of this algorithm is bounded 
by $\big(qnD\kappabar(f)\big)^{\Oh(n)}$.   
\end{prop} 

\subsection{Homologically equivalent complexes}\label{se:HEC}

The master plan to compute the homology of $\Ap(f,\Phi)$ passes through 
computing a simplicial complex $\fkC$ homologically equivalent to
$\Ap(f,\Phi)$. In the basic case, that is, when 
$\Phi$ is purely conjunctive, 
the construction of $\fkC$ is based on Theorem~\ref{teo:bcl} below.
Recall that the {\em Hausdorff distance} between 
two nonempty compact sets $W,V\subseteq\bbR^{n+1}$ is given by
$$
   d_H(W,V):=\max\Big\{\max_{v\in V} d(W,v),
   \max_{w\in W} d(w,V)\Big\}
$$
where $d$ denotes Euclidean distance in $\bbR^{n+1}$. 
If either $V$ or $W$ is empty then we take $d_H(V,W):=\infty$. 
We denote by $B(x,r)$ the Euclidean open ball 
with center $x$ and radius~$r$. 
Moreover, for a closed subset $X\subseteq\bbS^n$, 
we define the {\em open $r$-neighborhood of $X$} in $\bbR^{n+1}$,
\begin{equation}\label{eq:U}
\mcU(X,r):=\bigcup_{x\in X} B(x,r).
\end{equation} 
Similarly, we denote by $\mcU_{\bbS}(X,r)$ 
the {\em open (spherical) $r$-neighborhood} of $X$ in $\bbS^n$,  
which is defined with respect to angular distance. 
Clearly, 
\begin{equation}\label{eq:USU}
\mcU_{\bbS}(X,r)\subseteq  \mcU(X,r) .
\end{equation}
The following result~\cite[Theorems~2.8 and~4.12]{BCL17} 
is a variant 
of a seminal result by Niyogi, Smale, and Weinberger~\cite{NiSmWe2008}. 

\begin{theo}[Basic Homotopy Witness Theorem]\label{teo:bcl}
Let $f\in\Hd[q]$ and $\phi$ be purely conjunctive.
Moreover, let $\mcX\subseteq\bbS^n$ be a closed subset 
and $\varepsilon>0$ be such that
$$
   3 d_H\big(\mcX, \Ap(f,\phi)\big) 
   < \varepsilon <\frac{1}{14 D^{\frac32}\kappabar(f)}.
$$ 
Then the inclusion $\Ap(f,\phi) \hookrightarrow \mcU(\mcX,\varepsilon)$ 
is a homotopy equivalence.\eproof
\end{theo}

For a finite $\mcX$ satisfying the hypothesis of 
Theorem~\ref{teo:bcl} we can take $\fkC$ to be the {\em \v{C}ech complex} 
$\cech{\varepsilon}{\mcX}$ 
associated to $(\mcX,\varepsilon)$ whose $k$-faces, we 
recall~\cite[p.~60]{edelsbrunner_harer:2010}, 
are the sets of $k+1$ points $\{x_0,\ldots,x_k\}\subseteq\mcX$ such that
$\cap_{i\leq k} B(x_i,\varepsilon)\neq \varnothing$. The Nerve
Theorem~\cite[Corollary~4G.3]{hatcher} then states that the simplicial complex~$\fkC$ is
homotopically equivalent to $\mcU(\mcX,\varepsilon)$ and, 
by Theorem~\ref{teo:bcl}, 
to $\Ap(f,\phi)$ as well.  To compute the homology of $\Ap(f,\phi)$ it is therefore 
enough to construct a pair $(\mcX,\varepsilon)$ satisfying the
inequalities in Theorem~\ref{teo:bcl}, then build the complex $\fkC$, and 
finally compute the homology of $\fkC$. 
 
The retraction in the proof of Theorem~\ref{teo:bcl} is a nearest-point 
retraction. As we saw in the Introduction, these retractions do not 
necessarily exist for arbitrary closed semialgebraic sets. This is why 
we won't attempt to obtain a complex homotopically equivalent to $\Ap(f,\Phi)$. 
Instead, we will show that an appropriate {\em Boolean combination of 
\v{C}ech complexes} achieves {\sl homological equivalence}. 
We briefly describe how this is done.

Fix a finite set of points $\mcG$ in $\bbS^n$ and $\varepsilon>0$. 
For $\mcX_1,\mcX_2\subseteq\mcG$ we define the {\em intersection}  
$\cech{\varepsilon}{\mcX_1}\cap \cech{\varepsilon}{\mcX_2}$ 
to be the simplicial complex whose $k$-faces are the sets of points 
$\{x_0,\ldots,x_k\}$ which are $k$-faces of both
$\cech{\varepsilon}{\mcX_1}$ and $\cech{\varepsilon}{\mcX_2}$. 
It is clear that  
$\cech{\varepsilon}{\mcX_1}\cap \cech{\varepsilon}{\mcX_2}
= \cech{\varepsilon}{\mcX_1\cap \mcX_2}$. 
Similarly, we define the {\em union} 
$\cech{\varepsilon}{\mcX_1}\cup \cech{\varepsilon}{\mcX_2}$ 
to be the simplicial complex whose $k$-faces are the sets of points 
$\{x_0,\ldots,x_k\}$ which are $k$-faces of either 
$\cech{\varepsilon}{\mcX_1}$ or $\cech{\varepsilon}{\mcX_2}$. 
We observe that, in contrast with the behavior for intersections, 
we now only have 
$\cech{\varepsilon}{\mcX_1}\cup \cech{\varepsilon}{\mcX_2}
\subseteq \cech{\varepsilon}{\mcX_1\cup \mcX_2}$. 
The union complex is not necessarily a \v{C}ech complex over a subset of $\mcG$. 

Given a lax formula $\Phi$ over $f\in\Hd[q]$ and finite sets 
$\mcX_j^{\leq},\mcX_j^{=},\mcX_j^{\geq}\subseteq\bbS^n$, for $j=1,\ldots,q$, 
associated to the $3q$ {\em atomic} sets 
\begin{align}\label{eq:atomic}
   S_j^\leq &:=\Ap(f_j\leq 0),\nonumber\\ 
   S_j^=&:=\Ap(f_j=0), \\ 
   S_j^\geq&:=\Ap(f_j\geq 0),\nonumber
\end{align}
we can then consider the simplicial complex 
$\Phi\left(\cech{\varepsilon}{\mcX_1^\leq},\ldots, 
\cech{\varepsilon}{\mcX_{q}^\geq}\!\right)$, recursively built from 
the $\cech{\varepsilon}{\mcX_j^{\propto_j}}$ in the same way $\Ap(f,\Phi)$ 
is built from the $S_j^{\propto_j}$. 

The following result is our extension of Theorem~\ref{teo:bcl}. 
We will will prove it in~\S\ref{sec:complexes}.

\begin{theo}[Homology Witness Theorem]\label{theo:lax-case} 
Let $f\in\Hd[q]$ and $\varepsilon>0$. 
Moreover, for $j=1,\ldots,q$, let 
$\mcX_j^\leq,\mcX_j^=,\mcX_j^\geq\subseteq\bbS^n$ be closed subsets
such that for all $j$, $\mcX_j^\leq\cap \mcX_j^\geq=\mcX_j^=$ and such that for all $J\subseteq\{1,\ldots,q\}$ and all
$\propto\in\{\leq,=,\geq\}^J$, we have 
$$
   3 d_H\big(\cap_{j\in J}\mcX_j^{\propto_j}, \cap_{j\in J}S_j^{\propto_j}\big) 
   \ \leq\ \varepsilon \ \leq\ \frac{1}{14 D^{\frac32}   \kappabar(f)}.
$$ 
Then, for all lax formulas $\Phi$ over $f$, 
the set $\Ap(f,\Phi)$ and the simplicial complex 
$$
 \fkC =\Phi\Big(\cech{\varepsilon}{\mcX_1^{\leq}},\cech{\varepsilon}{\mcX_1^{=}},\cech{\varepsilon}{\mcX_1^{\geq}},\ldots, 
 \cech{\varepsilon}{\mcX_q^{\leq}},\cech{\varepsilon}{\mcX_q^{=}},\cech{\varepsilon}{\mcX_q^{\geq}}\!\Big)
$$ 
have the same homology.
\end{theo}

\begin{remark}\label{rem:wfs}
The techniques used to prove Theorem~\ref{teo:bcl} rely on the notion of 
{\em reach} (or {\em feature size})
$\tau(X)$ of a closed set $X$ in Euclidean space~\cite[\S4]{Federer:59}.  
This so because a positive reach of $X$ guarantees a nearest-point retraction onto 
$X$ from sufficiently small neigborhoods of $X$~\cite[Prop.~7.1]{NiSmWe2008},  
and a finite $\kappabar(f)$ guarantees a positive reach of 
$\Ap(f,\phi)$ for all purely conjunctive $\phi$~\cite[Thm.~4.12]{BCL17}. 
Unfortunately, for sets $X$ as the one drawn in the 
Introduction, the reach is zero. 
One may think that the use of the weak feature size $\tau_w$ and its associated retractions (see~\cite{ChCoLi:09,ChLi:05WFS}) 
could be an appropriate replacement of the reach for arbitrary closed semialgebraic sets,  
since $\tau_w$ is guaranteed to be positive on these sets. However, the proof 
of this positivity (see~\cite{ChLi:05axis}) does not give any effective way of 
bounding~$\tau_w$. Indeed, it is still an open problem to 
bound $\tau_w\left(\cup_{i=1}^mS_i\right)$ in terms of the $\tau_w(S_i)$ and of 
geometric quantities capturing the relative position of the $S_i$.
\end{remark}

\subsection{Algebraic neighborhoods of spherical semialgebraic sets}
\label{se:alg-neighborhoods}

Theorem~\ref{theo:lax-case} ensures that a collection of {\em point clouds} (finite sets 
of points) $\{\mcX_i^{\propto_i}\}$ sufficiently near to the sets $S_i^{\propto_i}$ 
allows us to build a simplicial complex homologically equivalent to $\Ap(f,\Phi)$. 
The difficulty we now face is, given a candidate set $\mcX_i^{\propto_i}$, 
how to estimate the Hausdorff distance between $\mcX_i^{\propto_i}$ and $S_i^{\propto_i}$. 
It was to solve this problem that algebraic neighborhoods were introduced 
in~\cite{BCL17}.

Algebraic neighborhoods of closed semialgebraic sets 
are obtained by relaxing the equalities and inequalities in their description.
More concretely, given a lax formula $\Phi$ over $f\in\Hd[q]$, 
the \textit{algebraic neighborhood $\Ap_r(f,\Phi)$ 
of $\Ap(f,\Phi)$ with tolerance $r$} is the spherical set defined by 
replacing the atoms $f_i=0$ by $|f_i(x)|\leq r \|f_i\|$, the atoms 
$f_i\geq 0$ by $f_i(x)\geq -r\|f_i\|$ 
and the atoms $f_i\leq 0$ by $f_i(x)\leq r\|f_i\|$. 
The \textit{open algebraic neighborhood $\Ap_r^\circ(f,\Phi)$ of $\Ap(f,\Phi)$ with 
tolerance $r$} is similarly defined but with strict inequalities. 

A crucial difference between $r$-neighborhoods and algebraic neighborhoods of 
$\Ap(f,\Phi)$ is that, for a given $x\in\bbS^n$, it is computationally trivial to check 
membership to the latter and computationally expensive to do so for the former. 
But to use algebraic neighborhoods to bound Euclidean distances we need to 
understand how do $r$-neighborhoods and algebraic neighborhoods  of $\Ap(f,\Phi)$ relate. 
The following inclusion is a simple consequence of the Exclusion Lemma (Lemma~3.1 
in~\cite{CKMW1}, see also~\cite[Prop.~4.17]{BCL17}) which goes back 
to~\cite{CS98} 
\begin{equation}\label{eq:Prop_c}
 \mcU_{\bbS}(\Ap(f,\Phi),r)\subseteq \Ap_{D^{1/2}r}^\circ(f,\Phi),
\end{equation}
where, we recall, $D=\max_i \deg f_i$. An inclusion in the other direction, now involving 
the condition of $f$, was shown in ~\cite[Thm.~4.19]{BCL17}. 

\begin{prop}\label{semialgebraiccaseneighborhoods}
Let $f\in\Hd[q]$ and $r>0$ be such that $13\,D^{3/2}\kappabar(f)^2 r<1$. 
Then, for every purely conjunctive formula $\phi$ over $f$, 
\begin{equation}\tag*{\qed}
 \Ap_r^\circ(f,\phi)\subseteq \mcU_{\bbS}(\Ap(f,\phi),3\,\kappabar(f)r).
\end{equation} 
\end{prop}

Proposition~\ref{semialgebraiccaseneighborhoods}, together with the choice of the 
$\mcX_i^{\propto_i}$ from a grid $\mcG$ sufficiently dense in $\bbS^n$, 
allows one to certify, in an 
efficient manner, that the hypothesis of Theorem~\ref{theo:lax-case} is satisfied. 
To prove this theorem, in addition, the following fundamental 
property of algebraic neighborhoods is used. 

\begin{theo}[Quantitative Durfee's Theorem]\label{semialgebraiccaseneighborhoodshomotopy}
Let $f\in\Hd[q]$ and $r>0$ be such that $\sqrt{2}\kappabar(f)r<1$. Then, 
for every purely conjunctive lax formula $\phi$ over $f$, the inclusions in
\begin{equation*}
  \begin{tikzcd}
  \Ap(f,\phi)\arrow[r,hook]  \arrow[dr,hook]
  & \Ap_r^\circ(f,\phi)\arrow[d,hook]\\
  & \Ap_{r}(f,\phi)
\end{tikzcd}
\end{equation*}
are homotopy equivalences.
\end{theo}

We will devote all of Section~\ref{se:geometry} to prove Theorem~\ref{semialgebraiccaseneighborhoodshomotopy}.

\subsection{Computation of homology groups}

Once in the possession of the complex $\fkC$ homologically equivalent to 
$\Ap(f,\Phi)$, the computation of the homology of the latter reduces to 
doing so for the former. Algorithmic procedures for this task are 
well-known. We briefly describe them (and recall their complexity) 
in~\S\ref{sec:groups}.

\subsection{Probabilistic estimates}

The primary complexity analysis of the algorithm {\sf Homology} is condition-based. 
Cost bounds on input $(p,\Phi)$ depend on the condition $\kappaff(p)$ (or on 
$\kappabar(f)$, as in the statement of Proposition~\ref{prop:kappa-est}). This 
complexity analysis takes form in part~(i) of Theorem~\ref{thm:main_result}.

To obtain parts~(ii) and (iii) one needs to estimate the probability tail 
of $\kappaff(p)$. To do so, we will bound $\kappaff(p)$ in terms of the 
normalized distance from $p$ to the set 
$\overline{\Sigma}^{\mathrm{aff}}_{\bfd}[q]
:=\sfH^{-1}(\overline{\Sigma}_{\bfd}[q])\subseteq\Pd[q]$  
of ill-posed tuples. 
This set is included in an algebraic cone~$V$ of codimension~1 
whose degree is bounded by an explicit function of $n,D$ and $q$. The tail 
$\prob_{p\in\Pd[q]}\{\kappaff(p)\ge t\}$ is consequently bounded in terms 
of the volume of the $\frac1t$-neighborhood of $V\cap\bbS(\Pd[q])$. 
A general result estimating 
this volume in terms of $N:=\dim\Pd[q]$, the degree of $V$, and its codimension  
is given in~\cite{BuCuLo:07}. We employ this result to estimate the tail 
of $\kappaff$ and use this estimate to obtain the desired weak cost bounds. 
This is carried out in Section~\ref{sec:final}.
\section{Condition and stability of the description}

Our algorithm's design and analysis are condition-based. To carry them 
out we will define an appropriate condition number, $\kappabar(f)$, and
show some of its main properties. This quantity follows a
lineage of condition numbers (for different problems)
going back to von Neumann and Goldstine~\cite{vNGo47} and Turing~\cite{Turing48}. 

We begin endowing the vector space $\Hd[q]$ with an inner product. 

\subsection{The Weyl inner product}\label{sec:Weyl}

For two homogeneous polynomials $g=\sum_{\alpha}g_\alpha X^\alpha$ and
$h=\sum_{\alpha}h_\alpha X^\alpha$ of degree $d$, the {\em Weyl inner
product} is given by
\begin{equation*}
 \langle g,h\rangle:=\sum_{\alpha}\binom{d}{\alpha}^{-1}g_\alpha h_\alpha,
\end{equation*}
where $\binom{d}{\alpha}:=\frac{d!}{\alpha_0!\cdots \alpha_d!}$ is the
multinomial coefficient. We extend this to pairs $f,f'\in\Hd[q]$ in
the usual way,
\begin{equation}\label{eq:Weyl2}
  \langle f,f'\rangle:=\sum_{i=1}^q\langle f_i,f_i'\rangle.
\end{equation}
The most important feature of this inner product is that it is
invariant under orthogonal changes of coordinates, i.e., that for each
$u\in \msO(n+1)$ and $f,g\in\mathcal{H}_d[q]$, $\langle
f,g\rangle=\langle f\circ u,g\circ u\rangle$. In addition, it is the
only inner product in $\Hd[q]$ invariant under orthogonal
transformations, up to renormalization in each component of $\Hd[q]$,
that extends to a Hermitian inner product invariant under unitary
transformations in the complex analog of $\Hd[q]$;
cf.~\cite[Rem.~16.4]{Condition}. 

The definition of the Weyl inner product in $\Hd[q]$ naturally 
translates to $\Pd[q]$: for $p,s\in\Pd[q]$ we define 
$\langle p,s\rangle:=\langle p\hm,s\hm\rangle$, where homogeneization 
is componentwise.

\subsection{The \texorpdfstring{$\mu$}{mu}-condition}\label{sec:mu}

We will look at elements $f\in\Hd[q]$ as polynomial maps
$f:\bbS^n\rightarrow \bbR^q$.  For $x\in \bbS^n$, we will denote by
$\diff_xf$ the tangent map $\Tg_x\bbS^n\rightarrow \bbR^q$. This
is nothing more than the restriction to the linear subspace
$\Tg_x\bbS^n$ of the usual derivative map of $f$ at $x$.

The \textit{$\mu$-condition} of $f$ at $x\in\bbS^n$ is given by
\begin{equation}
  \mu(f,x):=\|f\|\|\diff_xf^{\dagger}\Delta\|,
\end{equation}
where $\diff_xf^{\dagger}$ is the Moore-Penrose inverse of $\diff_xf$, 
$\Delta$ is the normalization matrix given by
\[
  \Delta:=\begin{pmatrix}
  \sqrt{d_1}&&\\&\ddots&\\&&\sqrt{d_q}
\end{pmatrix},
\]
and the norm $\|\diff_xf^{\dagger}\Delta\|$ is the spectral norm. By
convention, we take $\mu(f,x)$ to be $\infty$ when $\diff_xf$ is not
surjective.  One should see the inverse of $\mu(f,x)$ as a measure
of how near from being non-surjective the tangent map $\diff_xf$ is. 
The extra parameters $\|f\|$ and
$\Delta$ are there to ensure scalability as well as the equalities in
Theorem~\ref{homconditionthm}.

\begin{remark}
The $\mu$ condition number was introduced by Shub and Smale in their
``B\'ezout series"~\cite{Bez1,Bez2,Bez3,Bez4,Bez5}. It plays a crucial
role in the solution of Smale's 17th
problem~\cite{Smale98,BePa:09,BuCu11,Lairez17}.  The version
$\mu(f,x)$ slightly differs from the one in these references; it is
instead the minor variation introduced as $\mu_{\mathrm{proj}}$
in~\cite{BCL17} which allows, as shown in~\cite{BCL17}, an elegant
Condition Number Theorem (Theorem~\ref{homconditionthm} below).
\end{remark}

\subsection{The \texorpdfstring{$\kappa$}{kappa}-condition}

The quantity $\mu(f,x)$ is a good measure of how well-conditioned a
zero $x$ of $f$ is. A large value of $\mu(f,x)$ when $f(x)\neq 0$ 
indicates, when working over the complex numbers, that there exists 
a small perturbation of $f$ having an ill-posed zero. The fact that 
this is no longer true in the real case led 
to the introduction, in~\cite{Cucker99b}, of the following condition
number.

\begin{defi}\label{conddefi}
We define the 
{\em real homogeneous condition number 
of  $f\in\Hd[q]$ at  $x\in\bbS^n$} as 
\begin{equation*}
  \kappa(f,x):=\frac{1}{\sqrt{\frac{1}{\mu(f,x)^2}+\frac{\|f(x)\|^2}{\|f\|^2}}},
\end{equation*}
where we use the usual conventions of infinity together with 
$\infty^{-1}=0$ and its reciprocal. We further define 
the {\em real homogeneous condition number of $f$} by 
\begin{equation*}
\kappa(f):=\max_{x\in\bbS^n}\,\kappa(f,x).
\end{equation*}
\end{defi}

\begin{remark}
For $q>n$ and $f\in\Hd[q]$, the system $f=0$ is overdetermined. This
implies that $\diff_xf$ cannot be surjective at any $x\in\bbS^n$: for all
$x\in\bbS^n$, $\mu(f,x)=\infty$ and
$\kappa(f,x)=\frac{\|f\|}{\|f(x)\|}$. In particular,
$\kappa(f)<\infty$ if and only if $\Ap(f=0)$ is empty.
\end{remark}

The condition number $\kappa$ satisfies a Condition Number
Theorem. That is, its inverse tells us how near $f\in\Hd[q]$ is from
being ill-posed. To be precise, note that $\kappa(f)=\infty$ if and
only if there is some $x\in \Ap(f=0)$ such that $\diff_x f$ is not
surjective. This motivates to define the \textit{set of systems
ill-posed at} $x\in\bbS^n$ as
\begin{equation*}
  \Sigma_{\bfd}[q]_x:=
 \{f\in\Hd[q]\,|\,f(x)=0\text{ and }\diff_xf\text{ is not surjective}\},
\end{equation*}
and the \textit{set of ill-posed systems} as
$\Sigma_{\bfd}[q]:=\bigcup_{x\in\bbS^n}\Sigma_{\bfd}[q]_x$. One should
notice that $f\notin \Sigma_{\bfd}[q]$ if and only if 0 is a regular value of
$f$ which, by the Implicit Function Theorem, is enough to guarantee
that $\Ap(f=0)$ is smooth. The Condition Number Theorem for $\kappa$ 
is then the following.

\begin{theo}{\rm \cite[Theorem 2.19]{BCL17}}\label{homconditionthm}
For all $f\in\Hd[q]$ and $x\in\bbS^n$,
\[
  \kappa(f,x)=\frac{\|f\|}{d(f,\Sigma_{\bfd}[q]_x)}\quad\text{and}\quad
  \kappa(f)=\frac{\|f\|}{d(f,\Sigma_{\bfd}[q])} ,
\]
where $d$ is the distance induced by the Weyl inner product.\eproof
\end{theo}

\begin{cor}
For every $f\in\Hd[q]$ and $x\in \bbS^n$, $\kappa(f,x)\geq 1$. \eproof
\end{cor}

The following bound on $\mu$ in terms of $\kappa$ 
relates the values of $\mu$ and $\kappa$ near the zero set. 
It provides an important guarantee of the surjectivity of $\diff_xf$. 

\begin{prop}\label{boundamubykappa}
If $f\in\Hd[q]$, $x\in\bbS^n$ and $\sqrt{2}\kappa(f,x)\frac{\|f(x)\|}{\|f\|}<1$, then
\[\mu(f,x)\leq \sqrt{2}\kappa(f,x)\text{.}\]
Moreover, $\diff_xf$ is surjective.
\end{prop}

\begin{proof}
By the definition of $\kappa$,
\[
   \frac{1}{\kappa(f,x)^2}=\frac{1}{\mu(f,x)^2}+\frac{\|f(x)\|^2}{\|f\|^2}
   \leq 2\max\left\{\frac{1}{\mu(f,x)^2},\frac{\|f(x)\|^2}{\|f\|^2}\right\}\text{.}
\]
Since $2\frac{\|f(x)\|^2}{\|f\|^2}<\frac{1}{\kappa(f,x)^2}$ by hypothesis, we have 
$\max\left\{\frac{1}{\mu(f,x)^2},\frac{\|f(x)\|^2}{\|f\|^2}\right\}=\frac{1}{\mu(f,x)^2}$
and the desired inequality follows. 
Finally, we note that 
$\mu(f,x)$ is finite 
if and only if $\diff_xf^\dagger$ is defined if and only if $\diff_xf$ is surjective.
\end{proof}

\subsection{The intersection condition}\label{sec:kappabar}

Assume a perturbation of the coefficients of $f$ changes the topology 
of $\Ap(f,\Phi)$. Then, along the way in this perturbation, a singularity 
must occur in some boundary of $\Ap(f,\Phi)$. Because of this, it is in the 
description of the boundary pieces where the condition for computing this 
topology lies.

The Zariski closure of one such boundary piece is given by some
polynomial equalities.  We note though that, once we have $n+1$ such
equalities, the intersection will have to be empty to be well-posed,
and so, there is no need to consider intersections of more than
$n+1$ polynomials. This suggests the following definition.

\begin{defi}
Given $f\in\Hd[q]$, the \textit{real homogeneous intersection condition number} 
of~$f$ is defined as
\[
  \kappabar(f):=\max\left\{\kappa\left(f^L\right)\,|\,L\subseteq \{1,\ldots,q\},\,
   |L|\leq n+1\right\}
\]
where $f^L:=(f_i)_{i\in L}$.
\end{defi}

The following result explains the name
``intersection condition'' as it shows that, in some sense, the inverse of
$\kappabar$ measures how near are the intersections between the
hypersurfaces given by $f$ from being non-transversal.

\begin{theo}\label{kappabarfinitenesschar}
Let $f\in \Hd[q]$. Then $\kappabar(f)$ is finite if and only if $0$
is a regular value of each $f_i$, i.e., for every $x\in\Ap(f_i=0)$,
the map $\diff_xf_i:\Tg_x\bbS^n\rightarrow \mathbb{R}$ is surjective, and
any intersection between the smooth subvarieties $\Ap(f_i=0)$ is
transversal, i.e., for all $I\subseteq \{1,\ldots,q\}$ and all
$x\in\bigcap_{i\in I}\Ap(f_i=0)$,
\[
  \sum_{i\in I}\codim_{\Tg_x\bbS^n}\,\Tg_x\big(\Ap(f_i=0)\big)
  =\codim_{\Tg_x\bbS^n}\bigcap_{i\in I}\Tg_x\big(\Ap(f_i=0)\big).
\]
\end{theo}

\begin{proof}[Sketch of the proof]
It is clear that $\kappabar(f)$ is finite if and only if for every
$L\subseteq\{1,\ldots,q\}$ of size at most $n+1$, the map
$\diff_xf^L:\Tg_x\bbS^n\rightarrow \mathbb{R}^{L}$ is surjective for
each $x\in \Ap\left(\bigwedge_{i\in L}(f_i=0)\right)$. Now, this will
happen if and only if the hyperplanes $\ker \diff_xf_i$ of $\Tg_x\bbS^n$
intersect transversally, but this is exactly the claimed equality 
as $\ker \diff_xf_i=\Tg_x\Ap(f_i=0)$.
\end{proof}

\begin{remark}\label{rem:2conditions}
Consider a purely conjunctive formula 
\begin{equation}\label{eq:basic-formula}
  \bigwedge_{I\in I} (f_i=0) \wedge \bigwedge_{j\in J} (f_j\propto_j 0)
\end{equation}
where $\propto_j\in\{\leq,\geq\}$ and $I,J\subseteq\{1,\ldots,q\}$ with 
$I\cap J=\varnothing$. 
A condition number $\kappa_*(f_I,f_J)$, now depending on {\em both} $f$ and $\phi$ 
was defined in~\cite{BCL17} as follows,
$$
    \kappa_*(f_I,f_J):=\max_{\substack{L\subseteq J\\|L|\leq n-|I|+1}}
    \kappa(f_{I\cup L})
$$
where $f_I=(f_i)_{i\in I}$ and similarly for $f_J$ and $f_{I\cup L}$. 
It is immediate to verify that $\kappabar(f)=\kappa_*(\varnothing,f)$ and that  
\begin{equation*}
   \kappa_\ast(f_I,f_J)\leq \kappabar(f).
\end{equation*}
These relations allow us to use the bounds in~\cite{BCL17} replacing 
$\kappa_*$ by $\kappabar$ in them. 
\end{remark}

As a first application of the remark above we note that, although we don't 
have an exact Condition Number Theorem for 
$\kappabar$, we do have a bound. The following result is an immediate consequence 
of~\cite[Theorem~4.10]{BCL17}. 

\begin{theo}\label{boundkappabardistance}
For all $f\in\Hd[q]$ we have 
\[
   \kappabar(f)\leq \frac{\|f\|}{d\left(f,\overline{\Sigma}_{\bfd}[q]\right)},
\]
where
\[
   \overline{\Sigma}_{\bfd}[q]:=\bigcup\left\{\Sigma_{\bfd}[q]_L\,|\,
   L\subseteq \{1,\ldots,q\},\,|L|\leq n+1\right\}
\]
with
\[
  \Sigma_{\bfd}[q]_L:=
  \left\{f\in\Hd[q]\,|\,\exists \xi\in\displaystyle\bigcap_{i\in L}\Ap(f_i=0)\,:\, 
  \mathrm{rank}\,\diff_\xi f^L<|L|\right\}
\]
and $d$ is the distance induced by the inner product of $\Hd[q]$.
\eproof
\end{theo}

\section{Geometry}
\label{se:geometry}

The main goal of this section is to prove 
Quantitative Durfee's Theorem~\ref{semialgebraiccaseneighborhoodshomotopy}.

\subsection{Mather-Thom theory} 

Let us start with a motivation. 
Gradient retractions are central in Morse theory, where they are used 
to establish homotopy equivalences between fibers of Morse functions 
at pairs of regular values without critical values in between. More precisely, 
it is known that for a submersion $\alpha:\mcM\rightarrow I$
from a compact manifold $\mcM$ to an interval $I\subseteq \mathbb{R}$,
the gradient of~$\alpha$ induces a homotopy equivalence
$\alpha^{-1}(t)\subseteq \alpha^{-1}(J)$
for any subinterval $J\subseteq I$ and any $t\in J$. In more general
terms, but also using the gradient of $\alpha$ to prove it, this
translates into the following statement (a particular case of 
Ehresmann's Lemma): for a submersion
$\alpha:\mcM\rightarrow I$ from a compact manifold~$\mcM$ to an interval
$I\subseteq \bbR$, the map
$\alpha:\mcM\rightarrow I$
is a trivial fiber bundle. Recall that a {\em trivial fiber bundle} $\alpha:E\rightarrow B$ is a continuous map of topological spaces for which there is a subspace $F$ of $E$ (the {\em fiber}) and a homeomorphism $h: E\rightarrow F\times B$
such that the diagram
\[
  \begin{tikzcd}
  E \arrow[rr,"h"] \arrow[dr,"\alpha"']
  && F\times B \arrow[dl,"\pi_B"]\\
  & B&
\end{tikzcd}
\]
commutes. That is, $\alpha$ is a projection in disguise.

The extension of these results to a more general class of maps 
is part of the so-called stratified Mather-Thom theory~\cite{Mather2012}, 
which allows one to generalize the results above from 
smooth to semialgebraic, not necessarily smooth, maps. 
We next outline the main notions of this theory
(see also \cite{gibson}). 

The following definition generalizes the notion of a triangulation of $\mcM$, 
by allowing to decompose~$\mcM$ into more general pieces.

\begin{defi}\cite{Mather2012,gibson}\label{defiwhitney} 
A \textit{Whitney stratification} of a smooth manifold $\mcM$ of dimension~$m$ 
is a partition $\mcS$ of $\mcM$ into locally closed smooth submanifolds of~$\mcM$, called {\em strata}, such that:
\begin{itemize}
\item[F] (Locally finite)  Every $x\in \mcM$ has a neighborhood intersecting finitely many strata only.
\item[W] (Whitney's condition b) 
For every strata $\varsigma,\sigma\in\mcS$, 
every point $x\in\varsigma\cap \overline{\sigma}$, every sequence of points $\{x_{\ell}\}_{\ell\in\bbN}$ in $\varsigma$ converging to $x$, and every 
sequence of points 
$\{y_{\ell}\}_{\ell\in\bbN}$ in $\sigma$ converging to $x$,
we have that, in all local charts of $\mcM$ around $x$,  
\[
  \lim_{\ell\to\infty}\overline{x_{\ell},y_{\ell}}\subseteq
  \lim_{\ell\to\infty}\Tg_{y_\ell}\sigma,
\]
provided both limits exist. 
The inclusion should be interpreted in the local coordinates of the chart:
$\overline{x_{\ell},y_{\ell}}$ denotes the straight line joining $x_{\ell}$ and $y_{\ell}$, 
$\Tg_{y_{\ell}}\sigma$ denotes the affine plane 
tangent to $\sigma$ at $y_{\ell}$, 
and the limits are to be interpreted in the corresponding Grassmannians of~$\bbR^m$. 
\end{itemize}
\end{defi}

\begin{remark}
It is usual for the definition of Whitney stratification to include the so-called boundary condition which states that for every pair of strata $\varsigma,\sigma\in\mcS$, $\varsigma\cap \overline{\sigma}\ne\varnothing$ implies $\varsigma\subseteq \overline{\sigma}$. We omit it from the definition, because, as shown in~\cite[p.~16]{gibson}, this condition is not needed.
\end{remark}

A few comments are in order. 
Clearly, every smooth manifold $\mcM$ has the obvious Whitney stratification $\{\mcM\}$.
Also, according to \cite[Lemma 2.2]{Mather2012}, it is sufficient to check Whitney's 
condition~b in a local chart.
As for condition F, we won't mention it in what follows as we 
will only deal with finite stratifications. 

The necessity of Whitney's condition b is demonstrated by the following well-known example. 
Consider the stratification of $\mathbb{R}^2$ 
consisting of the point $\{0\}$, the smooth one-dimensional submanifold 
$$
 C:=\{(e^t\cos(t),e^t\sin(t))\,|\,t\in \bbR\},
$$ 
and the open subset $\sigma:=\bbR^2\setminus(\{0\}\cup C)$.
Note that $C$ is a logarithmic spiral and that the angle between 
$\overline{0,x}$ and $\Tg_xC$ is $\pi/4$ for all $x\in C$. 
This implies that 
\[
   \lim_{\ell\to\infty} \overline{0,y_\ell}\nsubseteq
   \lim_{\ell\to\infty}\Tg_{y_\ell}C
\]
for all sequences $\{y_\ell\}$ of points in $C$, whenever the two limits of lines exist. 
Therefore, Whitney's condition b is violated 
(take the constant sequence $\{0\}$ as the sequence $\{x_\ell\}$).
Indeed, the purpose of condition~b is to exclude wild variations  
such as the one of the logarithmic spiral when approaching 
the origin. 

Whitney stratifications are closed under various operations.

\begin{prop}\label{prop:productW}
{\rm \cite[Ch.~I, (1.2) and (1.4)]{gibson}}
Let $\mcW$ be a Whitney stratification of a smooth manifold $\mcM$. 
\begin{enumerate}
\item[{\rm (R)}] If $U$ is an open subset of $\mcM$, then $\mcW_{|U}:=\{\sigma\cap U\mid \sigma\cap U\ne \varnothing\}$ is a Whitney stratification of $U$.
\item[{\rm (P)}] If $\mcW'$ is a Whitney stratification of a smooth manifold $\mcM'$, then $\mcW\times \mcW':=\{\sigma\times \sigma'\mid \sigma\in\mcW,\,\sigma'\in\mcW'\}$ is a Whitney stratification of $\mcM\times\mcM'$.\eproof
\end{enumerate}
\end{prop}

Thom's first isotopy lemma~\cite{thom:69} 
extends Ehresmann's Lemma 
to maps $\alpha:\mcM\rightarrow \bbR^k$ 
that are in a way compatible with a Whitney stratification of~$\mcM$.  
Recall that a {\em proper map} is a continuous map for which the preimage of any 
compact set is compact. 

\begin{theo}[Thom's first isotopy lemma]
\label{thomfirstlemmaB}
Let $\mcM$ be a smooth manifold with a Whitney stratification~$\mcS$ 
and let $\alpha:\mcM\rightarrow \bbR^k$ be a continuous proper map such that: 
\begin{itemize}
\item for each stratum $\sigma\in\mcS$, there is an open neighborhood $U$ of $\overline{\sigma}$ and a smooth map $\varphi:U\rightarrow \bbR^k$ such that $\alpha_{|\sigma}=\varphi$,
\item 
for each stratum $\sigma\in\mcS$, $\alpha_{|\sigma}:\sigma\rightarrow \bbR^k$ is surjective,
\item 
for each stratum $\sigma\in\mcS$, $\alpha_{|\sigma}:\sigma\rightarrow \bbR^k$ is a smooth submersion.
\end{itemize}
Then $\alpha$ is a trivial fiber bundle. In particular, for all subsets $U,V\subseteq \bbR^k$, 
$\alpha^{-1}(U) \subseteq\alpha^{-1}(V)$ is a homotopy equivalence 
whenever $U \subseteq V$ is so.
\end{theo}

In the versions of Thom's first isotopy lemma we found 
in the literature,~\cite[Proposition~11.1]{Mather2012} 
and~\cite[Ch.~II, Theorem~5.2]{gibson}, the map $\alpha:\mcM\rightarrow \bbR^k$ is 
assumed to be smooth. We will show in Appendix~\ref{sec:proofThom} how 
Theorem~\ref{thomfirstlemmaB} follows from 
the statement in~\cite{gibson}.

\subsection{Semialgebraic Whitney stratifications of algebraic neighborhoods}

Our next result constructs a Whitney stratification for algebraic neighborhoods of 
basic semialgebraic sets that satisfies the hypothesis needed for applying Thom's first isotopy lemma. 

We define the {\em negative part} of $t\in\bbR$ to be $|t|_-:=\max\{-t,0\}$. 

\begin{prop}\label{lem:stratification}
Let $f\in\Hd[q]$, $\rho>0$ be such that $\sqrt{2}\kappabar(f)\rho<1$ and $\phi$ be 
the purely conjunctive lax formula
\[
   \phi\equiv \bigwedge_{i\in E} (f_i=0)\wedge 
   \bigwedge_{i\in P} (f_i\geq 0)
\]
with disjoint index sets $E,P$ such that $E\cup P=\{1,\ldots,q\}$.
Consider the open subset $\mcM :=\Ap_\rho^\circ(f,\phi)\setminus \Ap(f,\phi)$ 
of $\bbS^n$ and the continuous map
$\alpha:\mcM\rightarrow \bbR$ defined by 
\[
 \alpha(x):=\max\left\{\max_{i\in E}\frac{|f_i(x)|}{\|f_i\|},\,
 \max_{j\in P}\frac{|f_j(x)|_-}{\|f_j\|}\right\}.
\]
Finally, for $K\subseteq E$ and $L\subseteq P$ define  
\begin{equation}
S_{K,L}:=\left\{x\in\mcM\,\middle\vert\, 
\begin{array}{l}\forall i\in E,\,i \in K \Leftrightarrow \alpha(x)=|f_i(x)|/\|f_i\|\\[2pt] 
\forall j\in P,\, j\in L  \Leftrightarrow \alpha(x)=|f_j(x)|_-/\|f_j\|
\end{array}\right\} .
\end{equation}
Then the collection $\mcW:=\{S_{K,L} \mid S_{K,L} \ne \varnothing\}$ is a Whitney stratification 
of $\mcM$. Furthermore, for each stratum $S_{K,L}$,
\begin{enumerate}[(1)]
\item $S_{K,L}$ has codimension $|K|+|L|-1$,
\item $\alpha_{|S_{K,L}}$ is a smooth submersion, and
\item $\alpha(S_{K,L})=(0,\rho)$.
\end{enumerate}
\end{prop}

We observe that, for $x\in\bbS^n$ and $r\ge0$, $\alpha(x)\le r$ if and only if 
$x\in\Ap_r(f,\phi)$. In particular, $\Ap(f,\phi)$ is the zero set of $\alpha$. 

To avoid breaking the line of thought, we postpone the proof of 
Proposition~\ref{lem:stratification} to~\S\ref{se:WS} below 
and use it now to show the following, which is our last step before proving 
Theorem~\ref{semialgebraiccaseneighborhoodshomotopy}.

\begin{prop}\label{prop:retractionsemialgebraicnhoods}
Let $f\in\Hd[q]$ and $r>0$ be such that $\sqrt{2}\kappabar(f)r<1$ 
and $\phi$ be a purely conjunctive lax formula. 
Then for all $r'\in(0,r)$, the inclusions in
\begin{equation*}
  \begin{tikzcd}
  \Ap_{r'}(f,\phi)\arrow[r,hook]  \arrow[dr,hook]
  & \Ap_r^\circ(f,\phi)\arrow[d,hook]\\
  & \Ap_{r}(f,\phi)
\end{tikzcd}
\end{equation*}
are homotopy equivalences.
\end{prop}

\begin{proof} 
After permuting the polynomials, changing signs and eliminating non-occurring 
polynomials, we can assume that 
\[
   \phi\equiv \bigwedge_{i\in E} (f_i=0)\wedge \bigwedge_{i\in P} (f_i\geq 0),
\]
where the union $E\cup P=\{1,\ldots,q\}$ is disjoint,  
since these operations do not increase the value of $\kappabar(f)$.
Also, let $\rho>r$ be such that $\sqrt{2}\kappabar(f)\rho<1$.

Proposition~\ref{lem:stratification} gives us a Whitney stratification of 
$\mcM =\Ap^\circ_\rho(f,\phi)\setminus\Ap(f,\phi)$
on which $\alpha\colon\mcM\rightarrow (0,\rho)$ is a proper map 
satisfying the hypothesis of Theorem~\ref{thomfirstlemmaB}. Notice that the first hypothesis is satisfied because on each $\overline{S_{K,L}}$, $\alpha$ agrees with the absolute value of a polynomial, which is smooth as long it does not vanish, which is guaranteed by the fact that $\alpha$ only takes positive values on $\Omega$.
We can therefore use this theorem to deduce that 
$\alpha$ is a trivial fiber bundle. Let $F$ denote its fiber. 
Then there is a continuous map 
$\varphi:\mcM\to F$ such that 
$$
 h\colon\mcM\rightarrow F\times (0,\rho),\, x\mapsto
 (\varphi(x),\alpha(x))
$$ 
is a homeomorphism. 
Using this function, we see that the inclusion 
$\Ap_{r'}(f,\phi) \hookrightarrow\Ap_r(f,\phi)$
is a homotopy equivalence by the following continuous retraction 
\begin{align*}
 \eta:\Ap_r(f,\phi)\times [0,1]&\rightarrow \Ap_r(f,\phi)\\
(x,t)&\mapsto \begin{cases}
x&\text{if }x\in \Ap_{r'}(f,\phi)\\
h^{-1}(\varphi(x),tr'+(1-t)\alpha(x))&\text{otherwise.}
\end{cases}
\end{align*}
This restricts to a continuous retraction of $\Ap_r^\circ(f,\phi)$ onto $\Ap_{r'}(f,\phi)$, 
which shows also that the inclusion $\Ap_{r'}(f,\phi)\hookrightarrow\Ap_r^\circ(f,\phi)$ 
is a homotopy equivalence.

To show that the inclusion $\iota:\Ap_{r}^\circ(f,\phi) \hookrightarrow\Ap_r(f,\phi)$ is a homotopy equivalence, 
consider the retraction $\varrho:x\mapsto\eta(x,1)$ of $\Ap_r(f,\phi)$ onto $\Ap_{r'}(f,\phi)$ and its restriction $\varrho_\circ$ 
to a retraction of $\Ap_r^\circ(f,\phi)$ onto $\Ap_{r'}(f,\phi)$. We have shown above, 
using the map $\eta$, that $\varrho$ and $\varrho_\circ$ are homotopic 
to the identity maps of $\Ap_r(f,\phi)$ and $\Ap_r^\circ(f,\phi)$, respectively. 
Hence $\iota$ is a homotopy equivalence, 
because both $\varrho\circ \iota=\varrho_\circ$ and $\iota\circ\varrho=\varrho$ are homotopic to the corresponding identity maps.
\end{proof}

\begin{remark}
Notice that the proof of Proposition~\ref{prop:retractionsemialgebraicnhoods} cannot 
be extended to the case $r'=0$, directly proving 
Theorem~\ref{semialgebraiccaseneighborhoodshomotopy}, because the hypotheses of 
Thom's first isotopy lemma don't apply if we include the zero set of $\alpha$ 
inside $\mcM$. But we can now proceed with the proof of this theorem.
\end{remark}

\begin{proof}[Proof of Theorem~\ref{semialgebraiccaseneighborhoodshomotopy}]
By Proposition~\ref{prop:retractionsemialgebraicnhoods} the inclusion 
$\Ap_r^\circ(f,\phi)\hookrightarrow \Ap_r(f,\phi)$ is a homotopy equivalence. 
It is therefore enough to show that so is  
$\Ap(f,\phi)\hookrightarrow \Ap_r(f,\phi)$, as the third equivalence is a 
consequence of these two.  

By the Semialgebraic Triangulation Theorem~\cite[Theorem~9.2.1]{BCR:98}, 
$\Ap_r(f,\phi)$ has the structure of a CW complex of which $\Ap(f,\phi)$ is a subcomplex. 
Therefore, by \cite[Prop.~A.5.]{hatcher}, there is an open neighborhood $U$ 
satisfying that $\Ap(f,\phi)\subseteq U\subseteq \Ap_r(f,\phi)$ 
and that $\Ap(f,\phi)\hookrightarrow U$ is a homotopy equivalence. Notice that $U$ is 
open in the sphere, because we can assume that $U\subseteq \Ap_r^\circ(f,\phi)$ 
by choosing it sufficiently small.

Because the family $\{\Ap_\rho(f,\varphi)\}_{\rho\in(0,r)}$ is a descending family 
of compact sets satisfying that $\cap_{\rho\in(0,r)}\Ap_\rho(f,\varphi)=\Ap(f,\phi)$, 
there exists a sufficiently small $r'\in(0,r)$ such that 
$\Ap_{r'}(f,\phi)\subseteq U$. 
This gives us the following sequence of inclusions
\begin{equation*}\label{inclusionsequence}
  \Ap(f,\phi)\hookrightarrow\Ap_{r'}(f,\phi)\hookrightarrow U
  \hookrightarrow\Ap_{r}(f,\phi).
\end{equation*}
Passing to $k$th homotopy groups we obtain the  sequence 
of group homomorphisms 
\[\begin{tikzcd}
\pi_k\left(\Ap(f,\phi)\right)\arrow[r,"\alpha"]
&\pi_k\left(\Ap_{r'}(f,\phi)\right)\arrow[r,"\beta"]
&\pi_k\left(U\right)\arrow[r,"\gamma"]
&\pi_k\left(\Ap_{r}(f,\phi)\right)
\end{tikzcd}\]
where $\beta\circ\alpha$ is an isomorphism by the choice of $U$ 
and $\gamma\circ \beta$ is so due to 
Proposition~\ref{prop:retractionsemialgebraicnhoods}. 
It follows that $\alpha$, $\beta$ and $\gamma$ are isomorphisms and, hence, 
so is $\gamma\circ\beta\circ\alpha$.

We have thus shown that the inclusion $\Ap(f,\phi)\hookrightarrow \Ap_r(f,\phi)$ 
induces an isomorphism of homotopy groups. 
This translates to a homotopy equivalence by virtue of the Semialgebraic 
Triangulation Theorem~\cite[Theorem~9.2.1]{BCR:98} and 
Whitehead's Theorem \cite[Theorem~4.5]{hatcher}, which states that a continuous map of 
CW complexes that induces an isomorphism of homotopy groups is an homotopy equivalence. 
\end{proof}

\begin{remark}
In both Theorem~\ref{semialgebraiccaseneighborhoodshomotopy} and
Proposition~\ref{prop:retractionsemialgebraicnhoods} 
the inclusions (with the exception of $\Ap_r(f,\phi)\hookrightarrow \Ap_r^\circ(f,\phi)$) 
are actually deformation retractions. 
For Proposition~\ref{prop:retractionsemialgebraicnhoods}, this can be seen by 
modifying our proof; for Theorem~\ref{semialgebraiccaseneighborhoodshomotopy}, 
one can conclude using the stronger version of 
Whitehead's Theorem for subcomplexes~\cite[Theorem~4.5]{hatcher}.
\end{remark}

\subsection{Trivializing charts and semilinear stratifications}
\label{se:WS}

The goal of this subsection is to prove Proposition~\ref{lem:stratification}. 
The overall idea of the proof relies on two stepping stones. Firstly, to show that, 
at each point $x$ of $\mcM$ we can define a local chart for which the normalized
components $f_i /|f_i|$ of $f$ are the coordinate functions. Secondly, once with these local charts at hand, to show that the values taken by 
the normalized polynomials are enough to define the desired stratification. 
As these are values of coordinate functions, the resulting strata are semilinear. 

These stepping stones are dealt with, respectively, by the two lemmas below. 
We begin with a simple consequence of the Implicit Function Theorem. Recall, $g^S:=(g_i)_{i\in S}$.

\begin{lem}\label{lem:trivialcoordinates}
For given $f\in\Hd[q]$ put $g_i := f_i /|f_i|$. 
Fix $x\in\bbS^n$, and let $r>0$ be such that $\sqrt{2}\kappabar(f)r<1$.
We define the index set 
\[
  S:=\{i\in\{1,\ldots q\} \,\mid\, |g_i(x)|\leq r\} 
\]
and set 
$\bar{u} := g^S(x) \in\bbR^S$. 
Then $|S|\leq n$, and there exist an open neighborhood $O_x$ of $x$ 
in~$\bbS^n$ and $\varepsilon>0$ with the following properties: 

\begin{enumerate}
\item[{\rm (t1)}] We have $|g_i(y)|>r$ for all $i\notin S$ and all $y\in O_x$. 

\item[{\rm (t2)}] For all $i$ such that $g_i(x)\neq 0$, the sign of $g_i$ does not change on $O_x$.

\item[{\rm (t3)}] The set 
$\mcZ_x:=\{y\in O_x \mid f^S(y)=f^S(x)\}$ 
is a smooth submanifold of $\bbS^n$ of codimension~$|S|$,  
and there exists a diffeomorphism~$h$
such that the diagram
\[
\begin{tikzcd}
O_x\arrow[rr,"h"] \arrow[dr,"g^S"']
&& \mcZ_x\times B(\bar{u},\varepsilon) \arrow[dl,"\pi_{B}"]\\
& B\big(\bar{u},\varepsilon\big)&
\end{tikzcd},
\]
commutes (that is, for every $i\in S$, $g_i$ becomes a coordinate projection in the coordinates on $O_x$ given by $h$).  
\end{enumerate}
\end{lem}

\begin{proof}
Assume first that $S$ is nonempty.
Proposition~\ref{boundamubykappa} implies that 
$\diff_x f^S$ is surjective, since 
$\sqrt{2}\kappa(f^S)\frac{\|f^S(x)\|}{\|f^S\|}<1$.
So clearly $|S|\le m$.
Hence the derivative of the map $g^S$ at~$x$ %
is surjective as well. 
The Implicit Function Theorem implies the existence
of a diffeomorphism $h$ and a neighborhood $O_x$ satisfying~(t3) with $\mcZ_x$ smooth.  
By shrinking~$O_x$, we can guarantee the properties (t1) and (t2). 
Finally, the assertion is easily checked if $S$ is empty. 
\end{proof}

We will call the pair $(O_x,h)$ a \textit{trivializing chart at $x$}.  
We can describe a point  $y\in O_x$ by its {\em trivializing coordinates} 
$(z,u)\in\mcZ_x\times B(\bar{u},\varepsilon\big)$, where $u=(u_i)_{i\in S}$ and $h(y)=(z,u)$.
In these coordinates, the normalized polynomial $g_i=f_i/\|f_i\|$, for $i\in S$,
takes the form $(z,u)\mapsto u_i$.  

Our second stepping stone establishes a Whitney stratification in 
a combinatorial situation, where all the strata are semilinear.
The setting is as follows.
Recall that $|t|_-=\max\{-t,0\}$ is the negative part of $t\in\bbR$.

Let $S=I\cup J$ be a partition of a nonempty finite set $S$.  
We associate with this partition the finite union of open halfspaces 
\[
  \Omega := \bigcup_{i\in I} \{u\in\bbR^S \mid u_i \ne 0 \} \cup 
      \bigcup_{j\in J} \{u\in\bbR^S \mid u_j <0 \} .
\]
Consider the function $\alpha:\bbR^S\rightarrow [0,\infty)$ 
defined by 
\[
\quad 
  \alpha(u):=\max\left\{\max_{i\in I} |u_i|,\,
  \max_{j\in J}|u_k|_-\right\} 
\]
and write 
$K_{u}:=\{i\in I\mid |u_i|=\alpha(u)\}$, 
$L_{u}:=\{j\in J\mid |u_j|_-=\alpha(u)\}$  
for the set of indices, where at $u\in \Omega$, 
$\alpha$ attains the maximum over~$I$ and $J$, respectively. 
If we define 
\[
  \sigma_{K,L} :=\{u\in\Omega\mid K_{u}=K,\,L_{u}=L\} 
\]
for a pair of subsets $K\subseteq I$ and $L\subseteq J$, 
we see that 
$\{\sigma_{K,L} \mid \sigma_{K,L} \ne \varnothing\}$ 
is a partition of $\Omega$. 
Also, it is easy to check that 
\begin{equation}\label{eq:varsigma}
\sigma_{K,L}:=\left\{u\in\Omega\,\middle\vert\, 
\begin{array}{l}
\forall i\in I,\, i\in K \Leftrightarrow \alpha(u)=|u_i|\\[2pt]
\forall j\in J,\, j\in L \Leftrightarrow \alpha(u)=|u_j|_- 
\end{array}\right\}.
\end{equation}

\begin{lem}\label{lem:explicitstratification}
In the above setting, 
$\mcW:=\{\sigma_{K,L} \mid \sigma_{K,L} \ne \varnothing\}$ 
is a Whitney stratification of $\Omega$. 
Furthermore, for each stratum $\sigma_{K,L}$ in~$\mcW$: 
\begin{enumerate}[(1)]
\item $\sigma_{K,L}$ has codimension $|K|+|L|-1$,
\item $\alpha_{|\sigma_{K,L}}$ is a smooth submersion, and
\item if $\sigma_{K',L'}\subseteq \overline{\sigma_{K,L}}$, then 
$\alpha(\sigma_{K',L'}) \subseteq \alpha(\sigma_{K,L})$.
\end{enumerate}
\end{lem}

We postpone the proof of this lemma (which is a long sequence of elementary 
arguments) and proceed with the proof of Proposition~\ref{lem:stratification}. 

\begin{proof}[Proof of Proposition~\ref{lem:stratification}]
By the locality of the definition of Whitney stratifications, in order 
to prove that a family of subsets $\mcW$ is a Whitney stratification 
of a manifold $\mcM$, it is enough to show that for every point $x\in \mcM$, 
there is an open neighborhood $O_x$ of~$x$ in $\mcM$ 
such that $\mcW\cap O_x:=\{S\cap O_x\,|\,S\in\mcW\}$ is a Whitney stratification of $O_x$. 
This last statement will in turn be proved 
by exhibiting, for each $x\in\mcM$,  
a local chart in which 
we can apply Lemma~\ref{lem:explicitstratification}. 

Fix $x\in \mcM :=\Ap_\rho^\circ(f,\phi)\setminus \Ap(f,\phi)$.
By Lemma~\ref{lem:trivialcoordinates}, 
there is a trivializing neighborhood~$O_x$ with trivializing coordinates~$(z,u)$.
By shrinking $O_x$ if necessary, we can assume that $O_x\subseteq \mcM$. 
Let $B:=B\big(\bar{u},\varepsilon\big) \subseteq\bbR^S$ 
denote the open ball defined in Lemma~\ref{lem:trivialcoordinates}. 
In the coordinates~$(z,u)$, 
we are in the situation of Lemma~\ref{lem:explicitstratification} when we take 
$$
S=\{i\in\{1,\ldots q\} \,\mid\, |g_i(x)|\leq \rho\},\ 
I:=E\cap S,\ J:=P\cap S.
$$
Since the assertion is trivial if $S$ is empty, we assume $S\ne\varnothing$. 
Applying  Lemma~\ref{lem:explicitstratification}, 
we obtain the Whitney stratification $\{\sigma_{K,L} \mid \sigma_{K,L} \ne \varnothing\}$ of $\Omega$,
which induces the Whitney stratification $\{\sigma_{K,L}\cap B \mid \sigma_{K,L}\cap B \ne \varnothing\}$ of 
the open ball~$B$, by Proposition~\ref{prop:productW}(R). 
This in turn induces the product Whitney stratification 
$\{\mcZ_x \times(\sigma_{K,L}\cap B) \mid \sigma_{K,L} \cap B \ne \varnothing\}$ 
of $\mcZ_x \times B \simeq O_x$ by Proposition~\ref{prop:productW}(P). 
We now note that each 
$\mcZ_x\times (\sigma_{K,L} \cap B)$ corresponds to
$S_{K,L}\cap O_x$ in the local coordinates and hence, $\mcW$ is a Whitney stratification 
of $\mcM$ as claimed. 

Using Lemma~\ref{lem:explicitstratification} in the 
corresponding local coordinates, we deduce easily 
the assertions~(1) and~(2) of Proposition~\ref{lem:stratification}
from the corresponding parts of Lemma~\ref{lem:explicitstratification}, 
as well as the fact that 
$\alpha(\overline{S_{K,L}}) =\alpha(S_{K,L})$ 
from its part~(3).

It remains to prove the third assertion, that claims that $\alpha(S_{K,L}) = (0,\rho)$.  
Note first that the inclusion $\alpha(S_{K,L}) \subseteq (0,\rho)$ follows from 
the definitions of $\mcM$ and $\alpha$. 
The image $\alpha(S_{K,L})$ is open since, by part~(2), $\alpha_{|S_{K,L}}$ is a submersion. 
We show now that $\alpha(S_{K,L})$ is also closed in $(0,\rho)$. 
By the connectedness of the interval $(0,\rho)$, and since $S_{K,L}$ is nonempty, this 
will imply $\alpha(S_{K,L}) = (0,\rho)$.

So consider a sequence 
$\{p_n\}$ in $S_{K,L}$ such that $\{\alpha(p_n)\}$ converges to $\alpha_\infty \in (0,\rho)$.
By passing to a subsequence, we may assume that $\{p_n\}$ converges to 
some point $p_\infty$ in $\bbS^n$.
Then $p_\infty \in \mcM$ as 
$\alpha(p_\infty)=\alpha_\infty\in(0,\rho)$.
Therefore $\alpha_\infty \in \alpha(S_{K,L})$, since we already know that 
$\alpha(\overline{S_{K,L}}) =\alpha(S_{K,L})$. 
We have shown that $\alpha(S_{K,L})$ is indeed closed in $(0,\rho)$ 
and the proof is complete.
\end{proof}

\begin{proof}[Proof of Lemma~\ref{lem:explicitstratification}]
We already verified that 
$\{\sigma_{K,L} \mid \sigma_{K,L} \ne \varnothing\}$
is a partition of $\bbR^S$.  
We will use equations describing the different strata and their closures.
Let $U_i$ denote the variable corresponding to the coordinate function $u\mapsto u_i$. The set 
$\sigma_{K,L}$ can be described by the expression 
\begin{equation}\label{eq:omega}
\left(\vee_{i\in I}(U_i\neq 0)\right)\vee \left(\vee_{j\in J}(U_j<0)\right) ,
\end{equation}
ensuring that $u\in\Omega$, together with the (highly redundant) system
\begin{equation}\label{systemsigma}
\begin{cases}
|U_k|>|U_i|&(k\in K,\,i\in I\setminus K)\\
|U_\ell|_->|U_j|_-&(\ell\in L,\,j\in J\setminus L)\\
|U_k|=|U_{k'}|&(k,k'\in K)\\
|U_\ell|_-=|U_{\ell'}|_-&(\ell,\ell'\in L)\\
|U_k|=|U_\ell|_-&(k\in K,\,\ell\in L).
\end{cases}
\end{equation}
Let $\Omega_{K,L}$ denote the open subset of $\Omega$,
described by~\eqref{eq:omega}, together with
$$
\begin{cases}
|U_k|>|U_i|&(k\in K,\,i\in I\setminus K)\\
|U_\ell|_->|U_j|_-&(\ell\in L,\,j\in J\setminus L).\\
\end{cases}
$$
We can obtain $\sigma_{K,L}$ as the intersection of~$\Omega_{K,L}$
with the union,  
over all $\xi\in\{-1,1\}^K$, 
of the linear subspaces of $\bbR^S$ 
given by 
\[
\begin{cases}
\xi_kU_k=\xi_{k'}U_{k'}&(k,k'\in K)\\
U_\ell=U_{\ell'}&(\ell,\ell'\in L)\\
\xi_kU_k=-U_\ell&(k\in K,\,\ell\in L) . 
\end{cases}
\]
Each of these linear subspaces has codimension $|K|+|L|-1$ since,
in order to get a minimal system of equations, 
we only need to select a variable and, for each of the remaing $|K|+|L|-1$ 
variables, keep an equation determining its value. 

To prove that $\sigma_{K,L}$ is a smooth submanifold of $\bbR^S$, 
it suffices to show that any point~$u$ lying in 
two of these linear subspaces necessarily lies outside of $\Omega$.
Indeed, suppose $u$ is such a point. Then there exists $a\in K$ such that 
$u_a=-u_a$ and hence $u_a=0$. This implies 
$u_i=0$ for all $i\in I$ and then 
$u_j=0$ for all $j\in J$. Therefore $u\not\in\Omega$. 
Thus $\sigma_{K,L}$ is indeed a locally closed smooth submanifold of $\bbR^S$.
In particular, we have shown part~(1) of Lemma~\ref{lem:explicitstratification}. 

For part~(2), we observe that the restriction of $\alpha$ to 
each of the linear subspaces that make~$\sigma_{K,L}$ agrees 
with either $\xi_k U_k$ for $k\in K$ or $-U_\ell$ for $\ell\in L$, which are 
non-zero linear maps on $\sigma_{K,L}$, as $\alpha$ does not take the value zero 
in $\Omega$. 

We now claim that the following three conditions are equivalent:
\begin{eqnarray}\label{eq:3equiv}
\mbox{(i)} &~& \sigma_{K',L'} \subseteq \overline{\sigma_{K,L}}\nonumber\\ 
\mbox{(ii)} &~& \sigma_{K',L'}\cap \overline{\sigma_{K,L}}\ne \varnothing\\
\mbox{(iii)} &~& K\subseteq K' \mbox{ and } L\subseteq L'.\nonumber
\end{eqnarray}
To show this equivalence, we first observe that the closure $\overline{\sigma_{K,L}}$ in $\Omega$ is described 
by~\eqref{eq:omega} together with the system 
obtained from~\eqref{systemsigma} 
by replacing the strict inequalities by lax inequalities.
This description shows the implication (iii) $\Rightarrow$ (i). 
The implication (i) $\Rightarrow$ (ii) is trivial.
We show now (ii) $\Rightarrow$ (iii) by contraposition.  
Suppose $K\not \subseteq K'$ and let $a\in K\setminus K'$. 
As $a\in K$,
\[\sigma_{K,L}\subseteq \{u\mid\alpha(u)=|u_a|\},\]
which implies  
\[\overline{\sigma_{K,L}}\subseteq \{u\mid\alpha(u)=|u_a|\}.\]
Moreover, as $a\not \in K'$, 
\[\sigma_{K',L'}\subseteq \{u\mid\alpha(u)>|u_a|\}.\]
Thus $\overline{\sigma_{K,L}}\cap \sigma_{K',L'} =\varnothing$.
The case $L\nsubseteq L'$ is shown in a similar way. 
So we have proved the equivalence of the three statements. 

To prove part~(3) it is enough to show that 
if $\sigma_{K',L'}\subseteq \overline{\sigma_{K,L}}$, 
every point $u\in \sigma_{K',L'}$ 
can be obtained as a limit of a 
sequence $\{u_n\}$ of points of $\sigma_{K,L}$ with the same 
value under~$\alpha$. 
By the above characterization, 
$\sigma_{K',L'}\subseteq\overline{\sigma_{K,L}}$ 
implies 
$K\subseteq K'$ and $L\subseteq L'$. This allows one to obtain the desired 
sequence to approach any point $u$ in $\sigma_{K',L'}$ 
by slightly varying only the components $u_t$ with 
$t\in (K'\setminus K)\cup (L'\setminus L)$; for example, we may take 
$(u_n)_t:=(1-1/n)u_t$, which is in $\sigma_{K,L}$ as it 
satisfies~\eqref{eq:omega} and~\eqref{systemsigma}. 

We finally show Whitney's condition b, thus completing the proof that $\mcW$ 
is a Whitney stratification. 
The tangent space 
$\Tg_x\sigma_{K,L}$ at a point $u\in \sigma_{K,L}$ is the linear subspace 
given by
\begin{equation}\label{eq:tanSpace}
\begin{cases}
 \sgn(u_k)U_k=\sgn(u_{k'})U_{k'}&(k,k'\in K)\\
  U_\ell=U_{\ell'}&(\ell,\ell'\in L)\\
  \sgn(u_k)U_k=-U_\ell&(k\in K,\,\ell\in L)
\end{cases}
\end{equation}
where $\sgn:\bbR\to\{-1,0,1\}$ is the sign function. 
Now assume $\sigma_{K',L'}\cap\overline{\sigma_{K,L}}\ne\varnothing$ which, 
by~\eqref{eq:3equiv}, 
means that $K\subseteq K'$ and $L\subseteq L'$. Consider sequences of 
points $\{u_n\}$ and $\{u'_n\}$ in $\sigma_{K,L}$ and 
$\sigma_{K',L'}$, respectively, such that they both converge 
to $u\in \sigma_{K',L'}\cap\overline{\sigma_{K,L}}$. 
By the definition of convergence we have that, for all $n$ large enough and 
$k\in K'$,
\[
  \sgn(u_k)=\sgn\left((u_n)_k\right)=\sgn\left((u'_n)_k\right)
\]
as $u_k\neq 0$ for all $k\in K'$. 
This implies that for all $n$ large enough, the line $\overline{u_nu'_n}$ through $u_n$ and 
$u'_n$ lies inside $\Tg_{u_n}\sigma_{K,L}$ as, 
by the equations~\eqref{eq:tanSpace} and the inclusions $K\subseteq K'$ and $L\subseteq L'$, 
both $u_n$ and $u'_n$ lie in $\Tg_{u_n}\sigma_{K,L}\subseteq \Tg_{u'_n}\sigma_{K',L'}$. 
As this inclusion is preserved in the limit, we see that Whitney's condition b holds.
\end{proof}


\section{Topology}

In this section, we introduce two tools to construct isomorphisms of homology groups: 
an Explicit Homological Nerve Theorem for \v{C}ech complexes and a Homological 
Inclusion-Exclusion Transfer. 
These tools will combine topological information of basic semialgebraic sets to obtain 
such information for general semialgebraic sets.

\subsection{Explicit Homological Nerve Theorem}

Recall (from~\S\ref{se:HEC}) the definition, for a finite set of 
points $\mcX\subseteq \bbR^m$ and $\varepsilon>0$, of the 
\textit{\v{C}ech complex} of $\mcX$ of radius $\varepsilon$. 
By the Nerve Theorem~\cite[Corollary~4G.3]{hatcher}, 
the \v{C}ech complex $\cech{\varepsilon}{\mcX}$ 
is homotopically equivalent to  
the open $\varepsilon$-neighborhood $\mcU(\mcX,\varepsilon)$ around $\mcX$ 
defined in~\S\ref{se:HEC}.  
In particular, $\mcU(\mcX,\varepsilon)$ and $\cech{\varepsilon}{\mcX}$ 
have the same homology. 

We next exhibit a map that realizes this isomorphism in homology.

Consider the {\em free simplex} with vertex set $\mcX$, which is defined as 
the set 
\[
   \Delta^{\mcX}:=\left\{\sum_{x\in \mcX}t_x[x]\,\bigg|\,\text{for all }
   x\in\mcX,\,t_x\geq 0,\,\sum_{x\in\mcX}t_x=1\right\}
   \subseteq\bbR^\mcX
\]
formed by the formal convex combinations of the points of $\mcX$. 
Here we use the 
notation $[x]$ to distinguish the vertex $[x]$ in $\Delta^{\mcX}$ from the 
point $x\in\mcX\subseteq\bbR^m$. For $\sigma\in\cech{\varepsilon}{\mcX}$, the free 
simplex $\Delta^\sigma$ lies inside $\Delta^{\mcX}$ as a face and this correspondence 
is compatible with intersections in the sense that 
$\Delta^{\sigma\cap\sigma'}=\Delta^\sigma\cap\Delta^{\sigma'}$. This implies that 
by taking the union of all these faces, we get the following {\em realization} 
\[
   \left[\cech{\varepsilon}{\mcX}\right]:=\bigcup\left\{\Delta^\sigma\,\bigg|\,\sigma\in 
   \cech{\varepsilon}{\mcX}\right\} 
\]
of the abstract simplicial complex $\cech{\varepsilon}{\mcX}$ inside $\Delta^{\mcX}$.
In fact, this is 
the simplest geometric realization of the given abstract simplicial complex. 

Consider the affine map $\pi\colon\bbR^{\mcX}\to\bbR^m$ that sends the vertex~$[x]$ 
to the corresponding point~$x$. 
In other words,
$$
 \pi\left(\sum_{x\in \mcX}t_x[x]\right) = \sum_{x\in \mcX}t_xx .
$$
Clearly, $\pi$ maps the free simplex $\Delta^{\mcX}$ onto the the convex hull $\conv(\mcX)$ 
of $\mcX$ in~$\mathbb{R}^m$. 
The next lemma implies that $\pi$ maps the realization $[\cech{\varepsilon}{\mcX}]$ 
to  $\mcU(\mcX,\varepsilon)$.

\begin{lem}\label{lemnerve}
Let $\varepsilon>0$ and $\mcX\subseteq \bbR^m$  be a finite family of points. 
If $\bigcap_{x\in\mcX}B(x,\varepsilon)\neq \varnothing$, then 
$\conv(\mcX)\subseteq \mcU(\mcX,\varepsilon)$.
\end{lem}

\begin{proof}
Without loss of generality, by Carathéodory's Theorem~\cite[Proposition~1.15]{ziegler12}, 
we can assume that $\conv(\mcX)$ is a simplex. 
Suppose $\bigcap_{x\in\mcX}B(x,\varepsilon)\neq \varnothing$. 
For a nonempty $\sigma \subseteq \mcX$, take  
$p'\in \bigcap_{x\in \sigma}B(x,\varepsilon)$, and 
let $p_\sigma$ be the closest point to $p'$ in $\conv(\sigma)$.
Then $p_\sigma\in\bigcap_{x\in\sigma}B(x,\varepsilon)$.
By perturbing, we can assume that $p_\sigma$ lies in the relative interior of $\conv(\sigma)$.

We now consider the barycentric subdivision of $\conv(\mcX)$ with respect to the family of points $\{p_\sigma\mid \sigma\subseteq \mcX\}$, which is a barycentric subdivision where we take $p_\sigma$ instead of taking the centroid in the relative interior of each face $\sigma\subseteq \mcX$.
It is sufficient to show that 
$\conv(\Delta) \subseteq\mcU(\mcX,\varepsilon)$
for every maximal simplex of this subdivision.
Every such simplex $\Delta$ has the form 
$\conv(p_{\{x_1\}},p_{\{x_1,x_2\}},\ldots,p_\mcX)$, where $x_i\in\mcX$, 
so we have $p_{\{x_1,\ldots,x_a\}}\in \bigcap_{i=1}^aB(x_i,\varepsilon)\subseteq B(x_1,\varepsilon)$ for each of each of its vertices $p_{\{x_1,\ldots,x_a\}}$.
Therefore,  $\Delta\subseteq B(x_1,\varepsilon)\subseteq \mcU(\mcX,\varepsilon)$ by convexity.
\end{proof}

We can now state the Explicit Homological Nerve Theorem for \v{C}ech complexes.

\begin{theo}[Explicit Homological Nerve Theorem]\label{nervetheorem}
The restriction $\tilde{\pi}\colon [\cech{\varepsilon}{\mcX}] \rightarrow \mcU(\mcX,\varepsilon)$
of the affine map $\pi$ induces an isomorphism in homology: 
$$
 \tilde{\pi}_* \colon H_*\left(\left[\cech{\varepsilon}{\mcX}\right]\right)\to H_*(\mcU(\mcX,\varepsilon)) .
$$ 
\end{theo}

\begin{proof}
Let $\sigma\in\cech{\varepsilon}{\mcX}$. Then $\bigcap_{x\in\mcX}B(x,\varepsilon)\neq \varnothing$ and so, 
by Lemma ~\ref{lemnerve} applied to $\sigma$, $\conv(\sigma)\subseteq \mcU(\sigma,\varepsilon)\subseteq \mcU(\mcX,\varepsilon)$. 
As $\left[\cech{\varepsilon}{\mcX}\right]=\bigcup_{\sigma\in\cech{\varepsilon}{\mcX}}\Delta^\sigma$ and $\conv(\sigma)=\pi\left(\Delta^\sigma\right)$, it follows that 
$\pi\left(\left[\cech{\varepsilon}{\mcX}\right]\right)\subseteq \mcU(\mcX,\varepsilon)$. 
Thus $\pi$ is a continuous map $\left[\cech{\varepsilon}{\mcX}\right]\rightarrow \mcU(\mcX,\varepsilon)$. 
It only remains to prove that it induces an isomorphism in homology.

Let $\{\phi_x\}_{x\in \mcX}$ be a partition of unity in $\mcU(\mcX,\varepsilon)$ 
subordinate to $\{B(x,\varepsilon)\}_{x\in\mcX}$. 
That is, the $\phi_x$ are continuous maps $\phi_x:\mcU(\mcX,\varepsilon)\rightarrow [0,1]$ such that $\phi_x$ is zero outside $B(x,\varepsilon)$ 
and $\sum_{x\in\mcX}\phi_x=1$. (For example, we could take 
$\phi_x := \frac{\rho_x}{\sum_{x\in\mcX}\rho_x}$ 
with $\rho_x(p) := \max\{\varepsilon-\|p-x\|,0\}$.)
We define the continuous map
$$
  \varphi:\mcU(\mcX,\varepsilon)\rightarrow \left[\cech{\varepsilon}{\mcX}\right],\ 
   p\mapsto \sum_{x\in \mcX}\phi_x(p)[x]
$$
and will show that $\pi\circ \varphi$ is homotopic to the identity $\mathrm{id}_{\mcU(\mcX,\varepsilon)}$. 
To do so, consider the linear homotopy 
$$
  t\mapsto t(\pi\circ \varphi)+(1-t)\mathrm{id}_{\mcU(\mcX,\varepsilon)}
$$
between $\pi\circ \varphi$ and $\mathrm{id}_{\mcU(\mcX,\varepsilon)}$. 
To show that this linear homotopy restricts to a homotopy of functions $\mcU(\mcX,\varepsilon)\rightarrow\mcU(\mcX,\varepsilon)$, 
we only have to check that for every $p\in \mcU(x,\varepsilon)$, the segment $\left[\pi(\varphi(p)),p\right]$ is contained in $\mcU(x,\varepsilon)$. 

In order to check this, put  
$\mcX := \{x\in\mcX\,|\,\phi_x(p)\neq 0\}$ 
and note that 
$$
\pi(\varphi(p))=\sum_{x\in \mcX}\phi_x(p)x \in\conv(\mcX) .
$$ 
We have  
$p\in \bigcap_{x\in\mcX} B(x,\varepsilon)$ 
since $\phi_x(p)\neq 0$ implies $d(x,p) <\varepsilon$.
By Lemma ~\ref{lemnerve} we have 
$\conv(\mcX) \subseteq \mcU(\mcX,\varepsilon)$. 
So $\pi(\varphi(p)) \in \mcU(\mcX,\varepsilon)$. 
Hence there exists $\tilde{x}\in\mcX$ such that 
$\pi(\varphi(p)) \in B(\tilde{x},\varepsilon)$. 
Since also $p\in B(\tilde{x},\varepsilon)$, we have 
$[p,\pi(\varphi(p))]\subseteq B(\tilde{x},\varepsilon)\subseteq \mcU(\mcX,\varepsilon)$.

So we have shown that $\pi\circ \varphi$ is homotopic to the identity.
Therefore, 
$\pi_*:H_\ell\left(\left[\cech{\varepsilon}{\mcX}\right]\right)\rightarrow H_\ell\left(\mcU(\mcX,\varepsilon)\right)$ is an epimorphism for every $\ell$. 
Now, by the Nerve Theorem \cite[Corollary ~4G.3]{hatcher}, $H_\ell\left(\left[\cech{\varepsilon}{\mcX}\right]\right)$ and 
$H_\ell(\mcU(\mcX,\varepsilon))$ are isomorphic finitely generated abelian groups. 
We conclude that $\pi$ induces an isomorphism in homology, because 
a surjective homomorphism between isomorphic finitely generated abelian groups 
is an isomorphism \cite[Exercises~4.2(10)]{robinson_groups}.
\end{proof}

\begin{remark}
Theorem~\ref{inclusionexclusion} below gives an alternative way of proving 
Theorem~\ref{nervetheorem} without using the Nerve Theorem.
\end{remark}

\subsection{Homological Inclusion-Exclusion Transfer}

The title of the subsection refers to the idea of inferring information on the homology of a space~$X$ 
(or a map between spaces) from the homology of intersections of subspaces,  
akin to the combinatorial inclusion-exclusion principle. 

Let $X$ be a topological space and $C_\bullet (X)$ be its singular
chain complex.  For $A,B\subseteq X$ we denote by $C_\bullet(A+B)$ the
subcomplex of $C_\bullet(A\cup B)$ generated by the singular simplices
that either lie inside $A$ or inside $B$.  We will say that a finite
family $\{X_i\}_{i\in I}$ of subsets of~$X$ satisfies the
\textit{Mayer-Vietoris hypothesis} when, for every non-empty
$J\subseteq I$ and $k\in I\setminus J$, the inclusion of chain
complexes
\[
   C_\bullet\left(X_k + \bigcup_{j\in J}X_j\right)\hookrightarrow C_\bullet
   \left(X_k \cup \bigcup_{j\in J}X_j\right),
\]
induces an isomorphism in homology. We will say that it satisfies the
\textit{inductive Mayer-Vietoris hypothesis} when, for all finite families 
$\{F_\ell\}_{\ell\in L}$ of subsets of $I$,  
the family of intersections 
$\{\cap_{h\in F_\ell}X_h\}_{\ell\in L}$  
satisfies the Mayer-Vietoris hypothesis.

The reason to introduce this last notion is that it gives a common name to the three 
main situations that we will encounter and in which this condition holds:
\begin{enumerate}[1)]
\item 
The family $\{X_i\}_{i\in I}$ is a family of open subsets of
$\bigcup_{i\in I}X_i$. The inductive Mayer-Vietoris hypothesis 
holds due to~\cite[Proposition~2.21]{hatcher}.
\item 
The family $\{X_i\}_{i\in I}$ is a family of closed subcomplexes of
a CW-complex. The inductive Mayer-Vietoris
hypothesis holds due to \cite[Cor.~8.44]{rotman2}. 
\item 
The family $\{X_i\}_{i\in I}$ is a family of closed semialgebraic sets
in $\mathbb{R}^N$. The inductive Mayer-Vietoris hypothesis holds
due to the Semialgebraic Triangulation Theorem~\cite[Theorem~9.2.1]{BCR:98} 
combined with situation 2) above.
\end{enumerate}

In all these three situations, the inductive Mayer-Vietoris hypothesis 
will allow us to use the Mayer-Vietoris exact sequence in inductive arguments, 
such as the one for the following theorem.

\begin{theo}[Homological Inclusion-Exclusion Transfer]\label{inclusionexclusion}
Let $X$ and $Y$ be topological spaces and $\{X_i\}_{i\in I}$, $\{Y_i\}_{i\in I}$ be finite families of subsets 
of $X$ and $Y$, respectively, satisfying the inductive Mayer-Vietoris hypothesis. 
We assume that  $X=\bigcup_{i\in I} X_i$ and $Y=\bigcup_{i\in I} Y_i$. 
Moreover, 
let $f\colon X\rightarrow Y$ be a continuous map such that $f(X_i)\subseteq Y_i$ for all $i\in I$. 
Let $k$ be an integer such that for all nonempty $J\subseteq I$ with $|J|\leq k$, the morphism
\[
  H_\ell(f) \colon H_\ell\left(\cap_{j\in J}X_j\right)\rightarrow H_\ell\left(\cap_{j\in J}Y_j\right)
\]
is an isomorphism for $\ell<k$ and an epimorphism for $\ell=k$. Then
\[
  H_\ell(f) \colon H_\ell(X)\rightarrow H_\ell(Y)
\]
is an isomorphism for $\ell<k$ and an epimorphism for $\ell=k$.
\end{theo}

The following is an immediate consequence of Theorem~\ref{inclusionexclusion}. 

\begin{cor}\label{cor:IE}
Under the assumptions of Theorem~\ref{inclusionexclusion} if, 
for all nonempty $J\subseteq I$, 
$f\colon\cap_{j\in J}X_j\rightarrow \cap_{j\in J}Y_j$ 
induces an isomorphism in homology, then 
$f\colon X\rightarrow Y$ 
induces an isomorphism in homology.\eproof
\end{cor}

\begin{proof}[Proof of Theorem~\ref{inclusionexclusion}]
The proof is by induction on the size of $I$, for arbitrary $k$. 
The assertion is trivial when $I$ is a singleton.

Let $I=I'\cup\{i_0\}$ with $i_0\not\in I'$. 
By assumption, we have $f\left(\cup_{i\in I'}X_i\right)\subseteq  \cup_{i\in I'}Y_i$, 
$f(X_{i_0})\subseteq Y_{i_0}$ and $f\left(\cup_{i\in I'}(X_{i_0}\cap
  X_i)\right)\subseteq \cup_{i\in I'}(Y_{i_0} \cap Y_i)$. 
By induction hypothesis, 
the maps 
\[
 \beta_\ell^1:H_\ell(X_{i_0})\rightarrow H_\ell(Y_{i_0}) \text{ and }
  \beta_\ell^2:H_\ell\left(\cup_{i\in I'}X_i\right)\rightarrow H_\ell\left(\cup_{i\in I'}Y_i\right)
\]
induced by $f$ are isomorphisms for $\ell<k$ and epimorphisms for $\ell=k$, and the maps 
\[
  \alpha_\ell:H_\ell\left(\cup_{i\in I'}(X_{i_0}\cap X_i)\right)\subseteq 
  H_\ell\left(\cup_{i\in I'}(Y_{i_0}\cap Y_i)\right)
\]
are isomorphisms for $\ell<k-1$ and epimorphisms for $\ell=k-1$. 
Here we view $\cap_{j\in J} (X_{i_0}\cap X_j) = X_{i_0} \cap (\cap_{j\in J} X_j)$ 
as an intersection of $|J|+1$ subsets, for $J\subseteq I'$ with $|J| \le k-1$.
(Note that the inductive Mayer-Vietoris hypothesis is 
necessary to apply the induction step, as it guarantees that the families 
$\{X_{i_0}\cap X_j\}_{j\in J}$ and $\{Y_{i_0}\cap Y_j\}_{j\in J}$ satisfy the induction hypothesis;  
this is not necessarily the case with the Mayer-Vietoris hypothesis.)

The map of pairs $f:\left(\cup_{i\in I'} X_i, X_{i_0}\right)\rightarrow \left(\cup_{i\in I'}Y_i, Y_{i_0}\right)$, 
and the fact that these pairs satisfy the Mayer-Vietoris hypothesis,  
induce the commutative diagram of Mayer-Vietoris sequences shown in Figure~\ref{diagraminclusionexclusion}, 
where $\alpha_\ell$, $\beta_\ell$ and $\gamma_\ell$ are the maps in homology induced by $f$.

\begin{figure}[h]
\begin{tikzcd}
H_\ell\left(\cup_{i\in I'}\left(X_{i_0}\cap X_i\right)\right)\arrow[d]\arrow[r,"\alpha_\ell"] 
& H_\ell\left(\cup_{i\in I'}\left(Y_{i_0}\cap Y_i\right)\right)\arrow[d]\\
H_\ell(X_{i_0}) \oplus H_\ell\left(\cup_{i\in I'} X_i\right) \arrow[d]\arrow[r,"\beta_\ell"] 
& H_\ell(Y_{i_0}) \oplus H_\ell\left(\cup_{i\in I'} Y_i\right)\arrow[d]\\
H_\ell\left(X_{i_0} \cup (\cup_{i\in I'}X_i)\right) \arrow[r]\arrow[d,"\gamma_\ell"] 
& H_\ell\left(Y_{i_0} \cup (\cup_{i\in I'}Y_i)\right)\arrow[d]\\
H_{\ell-1}\left(\cup_{i\in I'}\left(X_{i_0}\cap X_i\right)\right)\arrow[d]\arrow[r,"\alpha_{\ell-1}"] 
& H_{\ell-1}\left(\cup_{i\in I'}\left(Y_{i_0}\cap Y_i\right)\right)\arrow[d]\\
H_{\ell-1}(X_{i_0}) \oplus H_{\ell-1}\left(\cup_{i\in I'} X_{i}\right)\arrow[r,"\beta_{\ell-1}"] 
& H_{\ell-1}(Y_{i_0}) \oplus H_{\ell-1}\left(\cup_{i\in I'} Y_{i}\right)
\end{tikzcd}
\centering
\caption{Natural maps of Mayer-Vietoris sequences in the proof of Theorem~\ref{inclusionexclusion}.}\label{diagraminclusionexclusion}
\end{figure}
%

In this figure, the induction hypothesis ensures that $\alpha_\ell$ is an 
isomorphism for $\ell< k-1$, an epimorphism for $\ell=k-1$, 
and that $\beta_\ell$ is an isomorphism for $\ell<k$ and an epimorphism for $\ell=k$. 
This gives us two cases to consider: $\ell\leq k-1$ and $\ell=k$.

If $\ell\leq k-1$, then $\alpha_{\ell-1}$, $\beta_{\ell-1}$ and $\beta_{\ell}$ are isomorphisms and 
$\alpha_{\ell}$ is an epimorphism. Therefore, by the Five Lemma~\cite[Proposition 2.72(iii)]{Rotman}, 
$\gamma_\ell$ is an isomorphism.

Otherwise, if $\ell=k$, then $\beta_{\ell}$ and $\alpha_{\ell-1}$ are epimorphisms, and 
$\beta_{\ell-1}$ is an isomorphism. Therefore, by the Four 
Lemma~\cite[Proposition 2.72(i)]{Rotman}, $\gamma_\ell$ is an epimorphism.

The statement now follows by induction.
\end{proof}


\begin{remark}
Theorem~\ref{inclusionexclusion} can be considered a homological version of the Vietoris-Begle Theorem~\cite[p.~344]{Spanier} for homology in terms of coverings. 
For example, one can see that for a locally trivial fibration $\pi\colon E\rightarrow B$ with $(k-1)$-connected fiber $F$, 
the homological inclusion-exclusion transfer implies the homological Vietoris-Begle Theorem 
since, for every trivializing open subset $U\subseteq B$,
$H_\ell(F\times U)\rightarrow H_\ell(U)$ is an isomorphism for $\ell<k$ and an epimorphism for $\ell=k$,
\end{remark}

\section{Algorithms}\label{sec:algorithms}

Our algorithm follows the steps described (with broad strokes) in 
Section~\ref{sec:overview}: 
\begin{description}
\item[(1)] We estimate the intersection condition $\kappabar(f)$. 
We do this in Algorithm~{\sc $\kappabar$-Estimate} in Subsection~\ref{sec:est-kappa}. 

\item[(2)] We construct clouds of points $\mcX^{\propto_j}_j$ approximating 
the atomic sets $S^{\propto_j}_j$, for $j\in\{1,\ldots,q\}$ and 
$\propto_j\in\{\le,=,\ge\}$, and satisfying 
that intersections of clouds approximate the corresponding intersections of 
sets. We use these clouds and our estimate on $\kappabar(f)$ 
to produce a simplicial complex $\fkC$ 
having the same homology as $\Ap(f,\Phi)$. This 
is Algorithm~{\sc Simplicial} in Subsection~\ref{sec:complexes}.

\item[(3)]
We computation the homology of $\fkC$. This is standard. But we recall 
the procedure and its complexity in Subsection~\ref{sec:groups}.
\end{description}
To do these computations, a sequence of grids on $\bbS^n$ is necessary. 
In this section we first describe the nature of these grids (and how to construct 
them) and then proceed to 
describe and analyse the complexity of the algorithms in the steps 
above. This complexity analysis is condition-based: the bounds are in 
terms of $\kappabar(f)$, in addition to the general parameters $n$, $q$ and~$D$.

\subsection{Grids}

The algorithm uses a sequence of grids on the sphere, both for estimating 
$\kappabar(f)$ and for sampling points on $\bbS^n$. 
These grids are simply constructed by 
projection onto the unit sphere of a uniform grid in the boundary of a cube. 
This sequence of grids has been used in~\cite{CS98,CKMW1,CKS16,BCL17} and 
its basic properties have been proved in these papers. 
We will therefore be concise in what follows.

For $\ell\in\bbN$, let $\mcG_{\ell}$ be the image on $\bbS^n$ under the 
projection $x\to \frac{x}{\|x\|}$ of the set of points $x\in\bbZ^{n+1}$ 
with $\|x\|_\infty=\lceil 2^\ell\sqrt{n}\rceil$. Further, let 
$r_\ell:=2^{-\ell}$. Then, 
\begin{equation}\label{eq:grid-size}
  |\mcG_\ell| = (n2^\ell)^{\Oh(n)}
\end{equation}
and
\begin{equation}\label{eq:grid-cover}
  \bbS^n\subseteq \bigcup_{x\in\mcG_\ell}B_{\bbS}(x,r_\ell) 
  \subseteq\bigcup_{x\in\mcG_\ell}B(x,r_\ell). 
\end{equation}
Note that the last implies that $d_H(\mcG_\ell,\bbS^n)\leq r_\ell$. We finally 
observe that, given $\ell$, the grid~$\mcG_\ell$ is easily computable. 

\subsection{Estimating the condition}\label{sec:est-kappa}

Recall the Definition~\ref{conddefi} of the 
real homogeneous condition number $\kappa(f)$ of $f\in\Hd[q]$.
We use the Lipschitz character of the inverse of $\kappa$ as a map on the sphere 
to estimate global bounds for $\kappa$ based on a finite number of 
point evaluations. 

\begin{lem}
Let $f\in\Hd[q]$, $\ell\in\bbN$, and
\[
  \kappa_\ell(f):=\max\left\{\kappa(f,x)\mid x\in \mcG_\ell\right\}.
\]
Then $\kappa_\ell(f) \le \kappa(f)$. 
Moreover, if $2D \kappa_\ell(f) r_\ell<1$, we have 
\[
  \kappa(f)\leq \frac{\kappa_\ell(f)}{1-2D \kappa_\ell(f) r_\ell} .
\]
\end{lem}

\begin{proof}
The first claimed inequality is trivial. To prove the second we recall that, 
by~\cite[Proposition 4.7]{BCL17}, the map
$
\bbS^n\rightarrow [0,1],\ x\mapsto \kappa(f,x)^{-1}
$
is $D$-Lipschitz continuous with respect to the Riemannian distance on 
$\bbS^n$, and so $2D$-Lipschitz with respect to the 
Euclidean distance on $\bbS^n$. 
Let $x_*\in\bbS^n$ be such that $\kappa(f)=\kappa(f,x_*)$. 
By the inclusions~\eqref{eq:grid-cover}, 
there exists $x\in\mcG_\ell$ such that $d(x,x_*)<r_\ell$. Using the Lipschitz 
property for the pair $(x,x_*)$ it follows that
\[
  \frac{1}{\kappa_\ell(f)} - 2D r_\ell \le \frac{1}{\kappa(f,x)} - 2D r_\ell
  \leq \frac{1}{\kappa(f,x_*)}=\frac{1}{\kappa(f)} .
\]
The desired inequality follows.
\end{proof}

We immediately derive analogous bounds for the real homogeneous intersection 
condition number. 

\begin{cor}\label{cor:est-kappa}
Let $f\in\Hd[q]$, $\ell\in\bbN$, and
\[
  \kappabar_\ell(f):=\max\left\{\kappa_\ell(f^L)\mid L\subseteq\{1,\ldots,q\},\,|L|
  \leq n+1\right\}.
\]
Then $\kappabar_\ell(f) \le \kappabar(f)$.
Moreover, if $2D \kappabar_\ell(f) r_\ell<1$, we have 
\begin{equation}\tag*{\qed}
  \kappabar(f)\leq \frac{\kappabar_\ell(f)}{1-2D \kappabar_\ell(f) r_\ell}.
\end{equation}
\end{cor}

Corollary~\ref{cor:est-kappa} motivates (and provides a proof of correctness for) 
the following algorithm.

\begin{minipage}[t]{0.9\textwidth}
\begin{algorithm*}[H]
\DontPrintSemicolon
\SetKwInOut{input}{Input}
\SetKwInOut{output}{Output}
\caption{\textsc{$\kappabar$-Estimate}}\label{alg:kappaestimation}
\input{$f\in\Hd[q]$\\
$\rho\in (0,1)$ \\
$B\in(0,\infty]$}
\hrulefill\\
$\ell\leftarrow 0$\;
\Repeat{$2D \sfK\, r_{\ell} \le \rho$ or $B\leq \sfK$}{
$\ell\leftarrow \ell+1$\;
$\sfK\leftarrow \max\{\kappa(f^L,x)\,|\,x\in \mathcal{G}_{\ell},\,|L|\leq n+1\}$\;
}
if $B\le \sfK$ \KwRet{\tt fail}\\
else \KwRet{$\sfK$}\\ 
\hrulefill\\
\output{{\tt fail} or $ \sfK \in (0,\infty)$}
\postcondition{If {\tt fail}, then $B\leq \sfK \leq \kappabar(f)$; 
otherwise  
$(1-\rho) \kappabar(f) \leq \sfK \leq \kappabar(f)$}
\end{algorithm*}
\end{minipage}
\medskip

\begin{prop}\label{prop:kappa-estgen}
Algorithm~{\sc $\kappabar$-estimate} is correct. 
Its cost 
on input $(f,\rho,B)$ is bounded by
\[
   \big( qnD\min\{B,\kappabar(f)\}\,\rho^{-1}\big)^{\Oh(n)}.
\]
\end{prop}

\begin{proof} 
The correctness follows from Corollary~\ref{cor:est-kappa} and the stopping criterion, 
noting that at each iteration we have
$\sfK=\kappabar_\ell(f) \le \kappabar(f)$. 

To prove the cost bound assume that, after $\ell$ iterations, we have 
\begin{equation}\label{eq:no-iter}
   \ell\geq \log_{2} \big(2D\msK \rho^{-1}\big), 
\end{equation}
where $\msK:=\min\{B,\sfK\}$. 
Then 
$r_\ell=2^{-\ell} \leq \frac{\rho}{D\msK}$. 
If $B>\sfK$ then $\msK=\sfK$, $r_\ell\leq\frac{\rho}{2D \sfK}$, and the algorithm halts. 
On the other hand, if $B\leq \sfK$, the algorithm halts as well. 
Thus we have shown that the algorithms halts after at most
$$
 T := \log_{2} \big(2D \min\{B,\kappabar(f)\}\rho^{-1}\big)
$$
iterations. 
At the $\ell$th iteration, where $\ell \le T$, the number of points in $\mcG_{\ell}$ is, 
by~\eqref{eq:grid-size}, bounded by 
\begin{equation}\label{eq:size-grid}
  (n2^\ell)^{\Oh(n)} = (n2^T)^{\Oh(n)} .
\end{equation}
For each point $x\in\mcG_{\ell}$ we compute the value of $\kappa(f^L,x)$ for 
at most $\sum_{i=1}^{\min\{q,n+1\}} \binom{q}{i}\leq (q+1)^{n+1}$ subsets $L$. 
And each of these computations can be done in 
$\Oh(N+n^3)$ operations (see~\cite[\S2.5]{Lairez17}) where we recall that $N=\dim\Hd[q]$. 
(Actually, we compute $\kappa$ up to a factor of~$\sqrt{n}$, but we will disregard this fact for simplicity.) 
It is easy to see that 
\begin{equation}\label{eq:N}
  N \leq (nD)^{\Oh(n)} ,
\end{equation}
from where it follows that each $\kappa(f^L,x)$ is computed 
with cost $(nD)^{\Oh(n)}$. 

Putting all the previous bounds together we bound the cost of the computation by 
\begin{equation*}
   T \, (n2^T)^{\Oh(n)} \, 
   (q+1)^{n+1} (nD)^{\Oh(n)}
   \leq \big( qnD\min\{B,\kappabar(f)\} \rho^{-1}\big)^{\Oh(n)} ,
\end{equation*}
which finishes the proof. 
\end{proof}

\begin{remark}
Algorithm~{\sc $\kappabar$-Estimate} estimates $\kappabar(f)$ up to a 
precision $\rho$ in finite time, 
provided this condition number is not too large 
(not much bigger than $B$). When $B=\infty$ is given as input, 
it estimates $\kappabar(f)$ up to this 
precision but its running time is not bounded. In particular, 
if $\kappabar(f)=\infty$, then the algorithm loops forever.
\end{remark}

\begin{proof}[Proof of Proposition \ref{prop:kappa-est}]
This is just a particular case of 
Proposition~\ref{prop:kappa-estgen} with $B=\infty$ 
and $\rho=0.01$.
\end{proof}

\subsection{Computation of simplicial complexes}\label{sec:complexes}

Given $f\in\Hd[q]$, a lax 
formula $\Phi$, and $\ell\in\bbN$, we define the finite {\em cloud of points} 
\begin{equation}\label{def:cloud}
\mcX_{\ell}(f,\Phi):=\Ap^\circ_{D^{1/2}r_\ell}(f,\Phi)\cap \mcG_\ell.
\end{equation}
In the special case of atomic formulas $f_j\propto_j0$, we will write 
$\mcX_j^{\propto_j}$ for the corresponding cloud. 
The following theorem gives sufficient conditions on $\ell$ and $\kappabar(f)$ 
for the clouds~$\mcX_\ell(f,\Phi)$ to approximate, 
as in the hypothesis of Theorem~\ref{theo:lax-case},
the sets $\Ap(f,\Phi)$.  

\begin{theo}[Sampling Theorem]
\label{thm:sampling}
Assume $f\in\Hd[q]$ and $\ell\in\bbN$ are such that 
$13\,D^{2} \kappabar(f)^2r_\ell<1$. 
Then, for every lax formula $\Phi$, we have 
$$
d_H\left(\mcX_\ell(f,\Phi),\Ap(f,\Phi)\right)<3\,D^{1/2} \kappabar(f)r_\ell.
$$ 
\end{theo}

\begin{proof}
Without loss of generality we can assume that $\Phi$ is in DNF, as 
this assumption does not change the underlying set. Furthermore, we can assume 
that $\Phi$ is basic due to the inequality
\[
  d_H(\cup_{i=1}^tA_i,\cup_{i=1}^tB_i)\leq \max_i\,d_H(A_i,B_i)
\]
of the Hausdorff distance and the fact that 
$\mcX_\ell(f,\Psi_0\vee\Psi_1)=\mcX_\ell(f,\Psi_0)\cup \mcX_\ell(f,\Psi_1)$.

By the construction of $\mcX_\ell$, 
\[
\mcX_\ell(f,\Phi)\subseteq \Ap^\circ_{D^{1/2}r_\ell}(f,\Phi)
\subseteq \mcU\big(\Ap(f,\Phi),3D^{1/2}\kappabar(f)r_\ell\big),
\]
the last by Proposition~\ref{semialgebraiccaseneighborhoods} and \eqref{eq:USU}.
By~\eqref{eq:grid-cover}, for all $x\in\Ap(f,\Phi)$, there is some $g_x\in\mcG_\ell$ 
such that $d_\bbS(x,g_x)\leq d(x,g_x)<r_\ell$. Thus, 
by \eqref{eq:Prop_c},  
$g_x\in \mcU_\bbS(\Ap(f,\Phi),r_\ell)\subseteq 
\Ap^\circ_{D^{1/2}r_\ell}(f,\Phi)$ 
and so $g_x\in \mcX_\ell(f,\Phi)$. Hence
\[
   \Ap(f,\Phi)\subseteq \mcU(\mcX_\ell(f,\Phi),r_\ell)\subseteq 
   \mcU\big(\mcX_\ell(f,\Phi),3D^{1/2}\kappabar(f)r_\ell\big),
\]
as $D\geq 1$ and $\kappabar(f)\geq 1$.
The inequality on the Hausdorff distance follows from the two 
inclusions above.
\end{proof}

We can now put together the Homology Witness Theorem~\ref{theo:lax-case} 
and the Sampling Theorem~\ref{thm:sampling}.  
The fundamental observation to make is that one only needs to   
sample from each of the $3q$ atomic 
sets associated with $f\in\Hd[q]$ defined in~\eqref{eq:atomic}. 
The following (trivial) identity
\begin{equation}\label{eq:inter}
   \mcX_\ell(f,\Phi)
   =\Phi\Big(\mcX_1^\leq,\mcX_1^=,
  \mcX_1^\geq,\ldots,\mcX_q^\leq,
  \mcX_q^=,\mcX_q^\geq\Big)
\end{equation}
allows us to obtain, for any lax formula $\Phi$, the 
cloud $\mcX_\ell(f,\Phi)$ by sampling from these atomic sets. 

\begin{prop}\label{prop:cechconstruction}
Let $f\in\Hd[q]$, $\varepsilon>0$, and $\ell\in\bbN$ be such that
\begin{equation*}
   9D^{1/2}\kappabar(f)r_\ell
   <\varepsilon<\frac{1}{14D^{3/2}\kappabar(f)} .
\end{equation*}
Then, for all lax formulas $\Phi$ over $f$, the semialgebraic set 
$\Ap(f,\Phi)\subseteq\bbS^n$ and the simplicial complex 
\begin{equation}\label{eq:defC}
  \fkC:=\Phi\Big(\cech{\varepsilon}{\mcX_1^\leq},\cech{\varepsilon}{\mcX_1^=},
  \cech{\varepsilon}{\mcX_1^\geq},\ldots,\cech{\varepsilon}{\mcX_q^\leq},
  \cech{\varepsilon}{\mcX_q^=},\cech{\varepsilon}{\mcX_q^\geq}\Big)
\end{equation}
have the same homology.
\end{prop}

\begin{proof}
This follows from the Homology Witness Theorem~\ref{theo:lax-case} applied to $f\in \Hd[q]$ 
and to the finite sets $\mcX_j^{\propto_j}$ associated to the atomic formulas $f_j\propto_j0$ 
via~\eqref{def:cloud}. For this, we need to check that 
\[
  3d_H\left(\mcX_\ell(f,\Phi),\Ap(f,\Phi)\right)<\varepsilon
  <\frac{1}{14D^{3/2}\kappabar(f)}.
\]
However, the right-hand inequality holds by assumption 
and the left-hand inequality follows from the Sampling Theorem~\ref{thm:sampling}
(it is immediate to check that $13\,D^{2} \kappabar(f)^2r_\ell<1$ follows 
from our hypothesis).
\end{proof}

We provide now the proof of the crucial Homology Witness Theorem~\ref{theo:lax-case}. 

\begin{proof}[Proof of Theorem~\ref{theo:lax-case}]
%
Without loss of generality, we can assume that $\Phi$ is in DNF, i.e., it is of the form $\bigvee_{i\in I}\phi_i$ with each $\phi_i$ purely conjunctive. We can further assume 
that no polynomial appears twice in any of the $\phi_i$. We can do these 
assumptions because they change neither the semialgebraic set $\Ap(f,\Phi)$ 
nor the simplicial complex $\fkC$ defined in~\eqref{eq:defC}.

We will use the Inclusion-Exclusion Transfer (Corollary~\ref{cor:IE}) 
to show that both $\Ap(f,\Phi)$ and $\fkC$ have the same homology 
as the algebraic neighborhood $\Ap_\rho^\circ(f,\Phi)$ 
for $\rho=6D^{1/2}\varepsilon$. 
Note that for this $\rho$ we have 
$$
  \sqrt{2}\kappabar(f)\rho=\sqrt{2}\kappabar(f)6D^{1/2}
  \varepsilon\le\frac{\sqrt{2}\kappabar(f)6D^{1/2}}
  {14 D^{3/2} \kappabar(f)}<1. 
$$
We can then use 
Theorem~\ref{semialgebraiccaseneighborhoodshomotopy} 
to deduce that, for all $J\subseteq I$, the inclusion
\[
 \bigcap_{j\in J}\Ap(f,\phi_j)=\Ap(f,\wedge_{j\in J}\phi_j)\subseteq 
 \Ap_{\rho}^\circ(f,\wedge_{j\in J}\phi_j)
 =\bigcap_{j\in J}\Ap_{\rho}^\circ(f,\phi_j)
\]
induces an isomorphism in homology. 
In addition, we have 
\[
 \bigcup_{i\in I}\Ap(f,\phi_i)=\Ap(f,\Phi)\text{~~and~~}
 \bigcup_{i\in I}\Ap_{\rho}^\circ(f,\phi_i)=\Ap_{\rho}^\circ(f,\Phi),
\]
so we can apply the Inclusion-Exclusion Transfer 
to the families 
$\{\Ap(f,\phi_i)\}_{i\in I}$ and 
$\{\Ap_{\rho}^\circ(f,\phi_i)\}_{i\in I}$ to deduce that the inclusion
\begin{equation}\label{eq:1homot}
   \Ap(f,\Phi)\hookrightarrow \Ap_{\rho}^\circ(f,\Phi)
\end{equation}
induces an isomorphism in homology. 
\smallskip

We now need to show that $\fkC$ and $\Ap_{\rho}^\circ(f,\Phi)$ have 
the same homology. 
To do so, for $J\subseteq \{1,\ldots,q\}$ and $\propto\in\{\leq,=,\geq\}^J$, 
we define the closed set 
$\mcX_J^{\propto}:=\cap_{j\in J}\mcX_j^{\propto_j}$. 
We also let $\psi_J^{\propto}:=\bigwedge_{j\in J}(f_j\propto_j0)$.  
By construction, we have 
\[
  \Ap(f,\psi_J^{\propto})=\bigcap_{j\in J}S_j^{\propto_j}
\]
where the $S_j^{\propto_j}$  are the $3q$ atomic sets associated with $f\in\Hd[q]$ defined in~\eqref{eq:atomic}.

We first prove that for all $z$ in the Euclidean neighborhood
$\mcU(\mcX_J^{\propto},\varepsilon)$, we have 
\begin{equation}\label{eq:inequalitytshow2}
d_\bbS\left(\frac{z}{\|z\|},\Ap(f,\psi_J^{\propto})\right) < 6\varepsilon. 
\end{equation}
Indeed, for all $y_0,y_1\in\bbS^n$,
\[
   d_\bbS(y_0,y_1)\leq \frac{\pi}{2}d(y_0,y_1)\leq 2d(y_0,y_1).
\]
Consequently,
\begin{align*}
d_\bbS\left(\frac{z}{\|z\|},\Ap(f,\psi_J^{\propto})\right)
&\leq 2d\left(\frac{z}{\|z\|},\Ap(f,\psi_J^{\propto})\right)\\
& < 2d\left(\frac{z}{\|z\|},z\right)+2d(z,\mcX_J^{\propto})
+2d_H(\mcX_J^{\propto},\Ap(f,\psi_J^{\propto}))\\
&\leq 2d\left(\frac{z}{\|z\|},z\right)+2d(z,\mcX_J^{\propto})+2\varepsilon\\
&= 2d(z,\bbS^n)+2d(z,\mcX_J^{\propto})+2\varepsilon\\
&\leq 4d(z,\mcX_J^{\propto})+2\varepsilon\\
&\leq 6\varepsilon,
\end{align*}
where the second line follows from the triangular inequality for the 
Hausdorff distance, the third 
one from $d_H(\mcX_J^{\propto},\Ap(f,\psi_J^{\propto}))<\varepsilon/3$, the fourth 
one from the fact that $\frac{z}{\|z\|}$ is the nearest point to $z$ in $\bbS^n$, 
the fifth one from $\mcX_J^{\propto}\subseteq \bbS^n$ and the sixth and last one from $z\in\mcU(\mcX_J^{\propto},\varepsilon)$. 
Hence we have shown~\eqref{eq:inequalitytshow2}. As the set $\mcU(\mcX_J^{\propto},\varepsilon)$ is not included in the sphere $\bbS^n$ it will be convenient to 
consider, for any set $S\subseteq\bbS^n$ the cone 
\[
 \widehat{S}:=\{\lambda x\mid\lambda>0,\,x\in S\}
\]
over the spherical set $S$. Note that the inclusion
\begin{equation}\label{eq:2homot}
   S\overset{\simeq}{\hookrightarrow} \hat{S}
\end{equation}
is a homotopy equivalence since the map
\[
 \widehat{S}\times[0,1]\rightarrow \widehat{S},\ 
 (p,t)\mapsto \frac{p}{(1-t)+t\|p\|_2}
\]
induces a continuous retraction of $\widehat{S}$ 
onto $S$. These two spaces thus have the same homology.
We will briefly write $\widehat{\mcU}$ and $\widehat{\Ap}$ to denote the cone over the 
corresponding neighborhoods. 
As a consequence of~\eqref{eq:inequalitytshow2} we deduce that 
\[
  \mcU(\mcX_J^{\propto},\varepsilon)\subseteq
  \widehat{\mcU}_{\bbS}(\Ap(f,\psi_J^{\propto}),6\varepsilon)
  \subseteq \widehat{\Ap}^\circ_\rho(f,\psi_J^{\propto})
\]
the last by~\eqref{eq:Prop_c} and the definition of $\rho$. 
We therefore have the inclusions
\begin{equation}\label{inclusionstoshow}
  \begin{tikzcd}
  \Ap(f,\psi_J^{\propto})  \arrow[r,hook]  \arrow[dr,hook]
  & \mcU(\mcX_J^{\propto},\varepsilon) \arrow[d,hook]\\
  & \widehat{\Ap}^\circ_\rho(f,\psi_J^{\propto})
\end{tikzcd}
\end{equation}
the horizontal arrow by hypothesis and the diagonal by composition.

We now note that $\Ap(f,\psi_J^{\propto})\hookrightarrow\mcU(\mcX_J^{\propto},\varepsilon)$
induces an isomorphism of homology by Theorem~\ref{teo:bcl} and that so does 
$\Ap(f,\psi_J^{\propto})\hookrightarrow\widehat{\Ap}^\circ_\rho(f,\psi_J^{\propto})$,
now by Theorem~\ref{semialgebraiccaseneighborhoodshomotopy} 
and~\eqref{eq:2homot}. 
This implies that the inclusion $\mcU(\mcX_J^{\propto},\varepsilon)
 \hookrightarrow \widehat{\Ap}^\circ_\rho(f,\psi_J^{\propto})$ induces 
 the isomorphism 
\begin{equation}\label{eq:3homol}
H_*\big(\mcU(\mcX_J^{\propto},\varepsilon)\big)\simeq 
H_*\big(\widehat{\Ap}^\circ_\rho(f,\psi_J^{\propto})\big).
\end{equation} 
Thus, the map
\[
   \pi:\cech{\varepsilon}{\psi_J^{\propto}}\rightarrow 
   \mcU(\mcX_J^{\propto},\varepsilon)
\]
defined in Theorem~\ref{nervetheorem} composed with the vertical arrow  
in~\eqref{inclusionstoshow} yields a map
\[
   \pi':\cech{\varepsilon}{\psi_J^{\propto}}\rightarrow 
   \widehat{\Ap}^\circ_\rho(f,\psi_J^{\propto})
\]
that induces an isomorphism in homology, by Theorem~\ref{nervetheorem} 
and~\eqref{eq:3homol}. As the 
$\psi_J^{\propto}$ cover all the purely conjunctive formulas, up to 
equivalence, we have shown that, for every purely conjunctive formula 
$\phi$, the map 
\[
  \pi:\cech{\varepsilon}{\phi}\rightarrow \widehat{\Ap}_\rho^\circ(f,\phi)
\]
from Theorem~\ref{nervetheorem} is well-defined, i.e., the image is 
contained in the codomain, and induces an isomorphism in homology.

We come back to the general case. Since
\[
\fkC=\bigcup_{i\in I}\cech{\varepsilon}{\phi_i}
\text{~~and~~}
\widehat{\Ap}_{\rho}^\circ(f,\Phi)=\bigcup_{i\in I}\widehat{\Ap}_{\rho}^\circ(f,\phi_i),
\]
the map
\[
  \pi:\fkC\rightarrow \widehat{\Ap}_\rho^\circ(f,\Phi)
\]
coming from Theorem~\ref{nervetheorem} is well-defined, as we can guarantee 
that the image is contained in the codomain by the previous paragraph. This map 
induces an isomorphism in homology, by the Inclusion-Exclusion Transfer 
(Corollary~\ref{cor:IE})
applied to the families $\{\cech{\varepsilon}{\phi_i}\}_{i\in I}$ and 
$\{\widehat{\Ap}_{\rho}^\circ(f,\phi_i)\}_{i\in I}$. This is so because, as we have 
just seen, the map $\pi$ induces an isomorphism in homology for purely conjunctive 
formulas, together with the equalities 
\[
  \bigcap_{j\in J}\cech{\varepsilon}{\phi_j}
  =\cech{\varepsilon}{\wedge_{j\in J}\phi_j}\text{~~and~~}
  \bigcap_{j\in J}\widehat{\Ap}_{\rho}^\circ(f,\phi_j)
  =\widehat{\Ap}_{\rho}^\circ(f,\wedge_{j\in J}\phi_j)
\]
for all $J\subseteq I$. Using~\eqref{eq:2homot} again we conclude that 
$\fkC$ and $\Ap_\rho^\circ(f,\Phi)$ have 
the same homology.

We can conclude as we have shown that both $\Ap(f,\Phi)$ and $\fkC$ have 
the same homology as $\Ap_\rho^\circ(f,\Phi)$ for the chosen $\rho$. 
\end{proof}

As a consequence of Proposition~\ref{prop:cechconstruction}, we may construct 
the desired simplicial complex~$\fkC$ from the complexes 
$\cech{\varepsilon}{\mcX_j^{\propto_j}}$ using the Boolean combination that yields 
$\Ap(f,\Phi)$ from the atoms $S_j^{\propto_j}$. 

\begin{minipage}[t]{0.9\textwidth}
\begin{algorithm*}[H]
\DontPrintSemicolon
\SetKwInOut{input}{Input}
\SetKwInOut{output}{Output}
\caption{\textsc{Simplicial}}\label{alg:simplicial}
\input{$f\in\Hd[q]$\\
Lax formula $\Phi$ over $f$\\
$\sfK\in [1,\infty)$
}
\precondition{$0.99\, \kappabar(f)\le \sfK \leq \kappabar(f) $}
\hrulefill

$\ell\leftarrow \lceil\log_2 200D^2 \sfK^2\rceil$\;
$\varepsilon\leftarrow \frac{1}{20\,D^{3/2} \sfK}$\;
\For{$j = 1,\ldots,q$}{
compute $\mcX_j^{\leq}$ and 
$\fkA_j^{\leq}\leftarrow \cech{\varepsilon}{\mcX_j^{\leq}}$\;
compute $\mcX_j^{=}$ and 
$\fkA_j^{=}\leftarrow \cech{\varepsilon}{\mcX_j^{=}}$\;
compute $\mcX_j^{\geq}$ and 
$\fkA_j^{\geq}\leftarrow \cech{\varepsilon}{\mcX_j^{\geq}}$\;
}
$\fkC\leftarrow \Phi\left(\fkA_1^{\leq},\fkA_1^{=},\fkA_1^{\geq},\ldots,\fkA_q^{\leq},\fkA_q^{=},\fkA_q^{\geq}\right)$\;

\KwRet{$\fkC$}\\
\hrulefill\\
\output{Simplicial complex $\fkC$}
\postcondition{$\fkC$ has the same homology as $\Ap(f,\Phi)$.}
\end{algorithm*}
\end{minipage}
\medskip

\begin{prop}\label{prop:comp-simplicial}
Algorithm~{\sc Simplicial} is correct. The cost of running it 
on input $(f,\Phi,\sfK)$ is bounded by 
$$
(q+\size(\Phi))\,\big(nD\kappabar(f)\big)^{\Oh(n^2)}.
$$
The number of faces of $\fkC$ is bounded by $(nD\kappabar(f))^{\Oh(n^2)}$. 
\end{prop}

\begin{proof}
It is straightforward to verify that the values the algorithm sets for 
$\ell$ and $\varepsilon$, along with the precondition on $\sfK$, guarantee 
that 
\begin{equation*}
   9D^{1/2}\kappabar(f)r_\ell
   <\varepsilon<\frac{1}{14D^{3/2}\kappabar(f)}.
\end{equation*}
Hence, the hypothesis of Proposition~\ref{prop:cechconstruction} 
are satisfied and the correctness of the algorithm follows. 

We next focus on complexity. The cost of computing an atomic cloud $\mcX_j^{\propto_j}$ 
is that of evaluating $f$ at a point in $\mcG_\ell$ times the number of points 
in $\mcG_\ell$. The latter is $\big(n2^\ell\big)^{\Oh(n)}$ by \eqref{eq:size-grid} 
and the former is $\Oh(N)=(nD)^{\Oh(n)}$ by~\eqref{eq:N}. Using that  
$\sfK\le\kappabar(f)$ it follows that we can compute one such atomic 
cloud with cost $(nD\kappabar(f))^{\Oh(n)}$. Multiplying by $3q$ 
we obtain the cost of computing all of them. 

The computation of the set $F_k$ of $k$-faces of the \v{C}ech complex 
$\cech{\varepsilon}{\mcX_j^{\propto_j}}$ takes time 
$|\mcX_j^{\propto_j}|^k\,k^{\Oh(n)}$ 
(see~\cite[Lemma~4.2]{CKS16}). As 
$|\mcX_j^{\propto_j}|\le |\mcG_\ell|=(n2^\ell)^{\Oh(n)}$ by~\eqref{eq:grid-size}, 
the computation of the sets $F_k$ for $k=0,\ldots,n$ (we are not interested 
in $k>n$ as $\dim S_j^{\propto_j}\leq n$) has cost 
$$
   \sum_{k=0}^n (n2^\ell)^{\Oh(kn)}k^{\Oh(n)} \ \le\ 
    (n2^\ell)^{\Oh(n^2)} \ \le\ 
   (nD\kappabar(f))^{\Oh(n^2)}.
$$
Multiplying by $q$ we obtain the cost of computing all the 
$\cech{\varepsilon}{\mcX_j^{\propto_j}}$. 

To compute the simplicial complex~$\fkC$ we compute, $\size(\Phi)-1$ many times, 
a union or an intersection of two (already computed) \v{C}ech complexes $\fkC_1$ 
and $\fkC_2$ (see \S\ref{se:HEC}). 
The cost of each of these computations is linear in the size of $\fkC_1$ and 
$\fkC_2$ and hence, this final step has cost bounded by 
$\size(\Phi)(nD\kappabar(f))^{\Oh(n^2)}$. The statement now 
follows.
\end{proof}

\subsection{Computation of homology groups}\label{sec:groups}

The final procedure to obtain the homology of $\Ap(f,\Phi)$ computes 
the homology of the simplicial complex $\fkC$ returned by 
{\sc Simplicial}. The matrices $M_k$ corresponding to the 
boundary maps $\delta_k:C_k\to C_{k-1}$ for $k=1,\ldots,n$, 
where $C_k$ is the free Abelian group generated by the $k$-faces, 
have entries in $\{-1,0,1\}$. The Betti numbers 
$\beta_0(\fkC),\ldots,\beta_{n-1}(\fkC)$, as well as the 
torsion coefficients, of $\fkC$ are computed from these 
matrices via the computation of their Smith normal 
form. A description of how this is 
done is in~\cite[Proposition~4.3]{CKS16} where the following 
cost bound is also proved. 
 
\begin{prop}\label{prop:cost-groups}
The homology groups $H_0(\fkC),\ldots,H_{n-1}(\fkC)$ of $\fkC$ are computed from 
the matrices $M_k$ with a cost bounded by $nF^{\Oh(k)}$, 
where $F$ is the maximum over~$k$ of the number of $k$-faces of $\fkC$. 
In the case where $\fkC$ is the simplicial complex returned by {\sc Simplicial} 
with input $(f,\Phi,\sfK)$, this total cost is
\begin{equation}\tag*{\qed}
   \big(nD\kappabar(f)\big)^{\Oh(n^2)}.
\end{equation}
\end{prop}

\section[Affine Condition, Random Data and Proof of the Main Result]{Affine Condition, Random Data\\\hspace*{7cm}and Proof of the Main Result}
\label{sec:final}

\subsection{Affine intersection condition}\label{sec:affine}

Proposition~\ref{generaltospherical} shows that, for well-posed 
tuples $p\in\Pd[q]$, homogeneization 
reduces the computation of homology groups of semialgebraic sets to the 
same computation for spherical semialgebraic sets. In what follows 
we deal with this last unproved result in our overview. We start by 
defining the condition number $\kappaff(p)$. 
The following example shows that taking $\kappabar(p\hm)$ with 
$p\hm=(p_1\hm,\ldots,p_q\hm)$ is not good enough. 

\begin{exam}\label{ex:parabola}
Consider the parabola $Y-X^2\in\mcP_{(2)}[1]$, whose homogenization gives 
the homogeneous polynomial $ZY-X^2$ for which we can easily check that 
$\kappabar(ZY-X^2)<\infty$, as zero is a regular value of this polynomial 
on the sphere. However, arbitrarily small perturbations of $Y-X^2$ inside 
$\mcP_{(2)}[1]$, e.g., those of the form 
$Y+\varepsilon Y^2-X^2$, can turn our description into that of an ellipse or a 
hyperbola, each of them having a topology different from that of a parabola.
\end{exam}

Example~\ref{ex:parabola} shows that $\kappabar(p\hm)$ alone does not capture 
ill-posedness. We note, however that for all $c>0$, $\kappabar(ZY-X^2,cZ)=\infty$ 
as the parabola and the hyperplane at infinity do not intersect transversally. 
Hence, a condition measure of the form $\kappabar(p\hm,cX_0)$ for some
$c>0$ would be a good measure. 
We have chosen this constant $c$ to be the norm $\|p\hm\|$ in our 
definition of $\Hm$ (cf.~\S\ref{sec:homog}). 
This choice makes possible to prove Theorem~\ref{boundkappabaraffine} below. 

\begin{defi}
The \textit{affine intersection condition number of $p\in\Pd[q]$} is 
defined as
\[
  \kappaff(p):=\kappabar(\Hm(p)).
\]
\end{defi}

The following result, extending Theorem~\ref{boundkappabardistance}, 
is an immediate consequence of~\cite[Proposition~4.16]{BCL17} 
and Remark~\ref{rem:2conditions}. 

\begin{theo}\label{boundkappabaraffine}
Let $\overline{\Sigma}^{\mathrm{aff}}_{\bfd}[q]:=\Hm^{-1}(\overline{\Sigma}_{\bfd}[q])$. 
For all $p\in\Pd[q]$ we have 
\begin{equation}\tag*{\qed}
 \kappaff(p)\leq 4D\frac{\|p\|}
 {d\left(p,\overline{\Sigma}^{\mathrm{aff}}_{\bfd}[q]\right)}.
\end{equation} 
\end{theo}

\begin{proof}[Proof of Proposition~\ref{generaltospherical}]
Let $\Phi^>$ denote the formula 
\[
   \Phi^>:=\Phi(p_1\hm,\ldots,p_q\hm)\wedge(\|p\|X_0> 0).
\]
It is enough to check that $\Ap(\Hm(p),\Phi\hm)$ and $\Ap(\Hm(p),\Phi^>)$ are 
homotopically equivalent, since 
it is well-known that $W(p,\Phi)$ is homeomorphic to $\Ap(\Hm(p),\Phi^>)$.
This will follow from showing that $\Ap(\Hm(p),\Phi^=)$ 
is collared in $\Ap(\Hm(p),\Phi\hm)$ where
\[
  \Phi^=:=\Phi(p_1\hm,\ldots,p_q\hm)\wedge(\|p\|X_0= 0),
\]
because then, by~\cite[Lemma~4.13]{BCL17}, 
$\Ap(\Hm(p),\Phi^>)=\Ap(\Hm(p),\Phi\hm)\setminus \Ap(\Hm(p),\Phi^=)$ 
would be homotopically equivalent to $\Ap(\Hm(p),\Phi\hm)$. Recall that a subset 
$B$ of a topological space $X$ is said to be {\em collared} in $X$ if there is 
a homeomorphism $h:[0,1)\times B\rightarrow V$ onto an open neighborhood $V$ 
of $B$ in $X$ such that $h(0,b)=b$ for all $b\in B$.
Because of Brown's Collaring Theorem~\cite{Brown62,Connelly71} 
it is enough to show that $\Ap(\Hm(p),\Phi^=)$ is locally collared in 
$\Ap(\Hm(p),\Phi\hm)$, i.e., that for every $x\in \Ap(\Hm(p),\Phi^=)$ there is an 
open neighborhood $O_x$ of $x$ such that $\Ap(\Hm(p),\Phi^=)\cap O_x$ is collared 
in $\Ap(\Hm(p),\Phi\hm)\cap O_x$. 

Fix $x\in \Ap(\Hm(p),\Phi^=)$. As $\kappabar(\Hm(p))=\kappaff(p)<\infty$ we can 
take $r>0$ such that $\sqrt{2}\,\kappabar(\Hm(p))r<1$ and apply Lemma~\ref{lem:trivialcoordinates} which 
guarantees the existence of a neighborhood $O_x$ of $x$ in $\bbS^n$ and   
trivializing coordinates around $x$ with respect to $(\Hm(p),r)$ on that 
neighborhood. On these coordinates 
we obtain formulas for $\Ap(\Hm(p),\Phi^=)\cap O_x$ and 
$\Ap(\Hm(p),\Phi\hm)\cap O_x$ by substituting $p\hm_j$ by $U_j$,
$\|p\|X_0$ by $U_0$, and the atoms of those polynomials having 
constant sign on $O_x$ by true or false appropriately. After doing 
this, we obtain a formula $\Xi$ over 
$(U_i)_{i\in S\setminus 0}$, where $S$ is as in 
Lemma~\ref{lem:trivialcoordinates},  
for which $\Ap(\Hm(p),\Phi^=)\cap O_x$ 
is described by
\[
  \Xi\wedge (U_0=0)
\]
and $\Ap(\Hm(p),\Phi\hm)\cap O_x$ by
\[
   \Xi\wedge (U_0\geq 0).
\]
From this, it follows that the map
\begin{eqnarray*}
   h:[0,1)\times \Ap(\Hm(p),\Phi^=)\cap O_x&\rightarrow& 
   \Ap_r^\circ(\Hm(p),\Phi^=)\cap\Ap(\Hm(p),\Phi\hm)\cap O_x\\
   (t,(z,u))&\mapsto& (z,u_0+rt,(u_i)_{i\in S\setminus 0})
\end{eqnarray*}
is a homeomorphism of $[0,1)\times \Ap(\Hm(p),\Phi^=)\cap O_x$ with
an open neighborhood
of $\Ap(\Hm(p),\Phi^=)\cap O_x$ 
inside 
$\Ap(\Hm(p),\Phi\hm)\cap O_x$
for $r$ sufficiently small, 
since altering $u_0$ does not affect whether~$\Xi$, in which $U_0$ does not appear, 
is true or not. Hence, $\Ap(\Hm(p),\Phi^=)\cap O_x$ is collared in 
$\Ap(\Hm(p),\Phi\hm)\cap O_x$ and the proof is complete.
\end{proof}

\subsection{Random tuples in \texorpdfstring{$\Pd[q]$}{Pd[q]} and \texorpdfstring{$\Hd[q]$}{Hd[q]}}

To obtain weak complexity estimates we endow the unit sphere $\bbS(\Pd[q])$ 
with the uniform measure. We observe that, as $\kappaff(p)=\kappaff(\lambda p)$ 
for all $\lambda>0$, the probability tail for $\kappaff(p)$ is the same 
no matter whether $p$ is taken from the uniform distribution on $\bbS(\Pd[q])$ 
or from the standard Gaussian distribution on $\Pd[q]$ with respect to the 
Weyl monomial basis $\big\{\binom{d_j}{\alpha}X^\alpha\big\}_{|\alpha|=d_j}$ for each $p_j$. 

For any of these two distributions and for a condition number of the form 
$\msC(a)=\frac{\|a\|}{d(a,\Sigma}$ where $\Sigma\subseteq\bbR^{p+1}$ is an algebraic 
cone defined as the zero set of a homogeneous polynomial~$h$, the main result 
in~\cite{BuCuLo:07} (see also~\cite[Theorem ~21.1]{Condition} for the bound below) 
gives estimates on the tail of $\msC$ in terms of the degree of $h$ and the 
dimension of the ambient space: for all $t\ge (2\deg(h)+1)/p$, 
\begin{equation}\label{eq:BuCuLo}
  \prob_{a\in\bbR^{p+1}}\left\{\frac{\|a\|}{d(a,\Sigma)}\geq t\right\} =
  \prob_{a\in\bbS^p}\left\{\frac{1}{d(a,\Sigma)}\geq t\right\}
  \le 11 \frac{p\deg(h)}{t}.
\end{equation}
This result was used in~\cite{CKS16} and subsequently 
in~\cite{BCL17} (in conjunction with Theorem~\ref{boundkappabaraffine}) 
to obtain bounds for the tail of $\kappa(f)$ and $\kappa^*(f,g)$. Our 
proof of the next result will be consequently succinct. 

\begin{prop}\label{prop:tail}
For all $t\ge \frac{n 2^{n+1} (q+1)^{n+1}D^n+1}{N-1}$,
$$
   \prob_{p\in\bbS^{N-1}}\left\{\kappaff(p)\geq t\right\}\leq 
   \frac{44D^{n+1}(N-1)n(2(q+1))^n}{t}.
$$
\end{prop}

\begin{proof}
The set $\overline{\Sigma}^{\mathrm{aff}}_{\bfd}[q]$ is contained in the 
zero set of a polynomial in $N$ variables of degree bounded by 
$n2^n(q+1)^{n+1}D^n$ by~\cite[Corollary~4.21]{BCL17} 
and Remark~\ref{rem:2conditions}. We now use Theorem~\ref{boundkappabaraffine}
and~\eqref{eq:BuCuLo}.
\end{proof}

\subsection{Proof of the Main Result}

We begin by exhibiting the algorithm {\sc Homology.}
\medskip

\begin{minipage}[t]{0.9\textwidth}
\begin{algorithm}[H]
\DontPrintSemicolon
\SetKwInOut{input}{Input}
\SetKwInOut{output}{Output}
\caption{\textsc{Homology}}\label{alg:homology}
\input{$p\in\Pd[q]$\\
Lax formula $\Phi$ over $p$
}
\hrulefill

$f\leftarrow \Hm(p)$\;
$\sfK\leftarrow \mbox{\sc $\kappabar$-Estimate}(f,0.01,\infty)$\;  
$\fkC\leftarrow \mbox{\sc Simplicial}(f,\Phi\hm,\sfK)$\;
compute the homology groups $H_*(\fkC)$ of $\fkC$\; 
\KwRet{$H_*(\fkC)$}\\
\hrulefill\\
\output{A sequence of groups $H_*=(H_0,\ldots,H_{n-1})$}
\postcondition{$H_*$ is the homology sequence of $W(p,\Phi)$.}
\end{algorithm}
\end{minipage}
\medskip

Its correctness is a trivial consequence of Propositions~\ref{generaltospherical}, 
\ref{prop:kappa-est} and~\ref{prop:comp-simplicial}. The last two, 
together with Proposition~\ref{prop:cost-groups}, yield the bound  
\begin{eqnarray*}
\cost(p,\Phi) &\le&  (qnD\kappabar(f))^{\Oh(n)}
  +(nD\kappabar(f))^{\Oh(n^2)}
  +(q+\size(\Phi))(nD\kappabar(f))^{\Oh(n^2)}\\
&\le& \size(\Phi) q^{\Oh(n)} (nD\kappabar(f))^{\Oh(n^2)}
\end{eqnarray*}
for the cost of the algorithm on input $(p,\Phi)$. This proves part~(i) of 
Theorem~\ref{thm:main_result}. 

For part~(ii) we take $t=(nqD)^{cn}$ with $c>0$ large enough so that 
the hypothesis of Proposition~\ref{prop:tail} holds. Then, that proposition guarantees 
that 
$$
  \prob_{p\in\bbS^{N-1}}\left\{\kappaff(p)\geq (nqD)^{cn}\right\}\leq 
   \frac{44D^{n+1}(N-1)n(2(q+1))^n}{(nqD)^{cn}}
   \leq (nqD)^{-n}
$$
the last as $N=(nD)^{\Oh(n)}$ by~\eqref{eq:N} and by choosing 
$c$ large enough. It follows that with probability at least 
$1-(nD)^{-n}$ we have $\kappaff(p)\leq (nD)^{cn}$ 
and hence, by part~(i), $\cost(p,\Phi)\leq \size(\Phi) q^{\Oh(n)}
(nD)^{\Oh(n^3)}$. 

Finally, to prove part~(iii), we take $t=2^{c\,\size(p,\Phi)}$. It is easy to see that 
we can choose $c$ large enough so that the hypothesis of 
Proposition~\ref{prop:tail} holds. Again, that proposition then
guarantees that 
$$
  \prob_{p\in\bbS^{N-1}}\left\{\kappaff(p)\geq 2^{c\,\size(p,\Phi)}\right\}\leq 
   \frac{44D^{n+1}(N-1)n(2(q+1))^n}{2^{c\,\size(p,\Phi)}}  
   \leq 2^{-\size(p,\Phi)}
$$
the last inequality by choosing $c$ large enough. As before, it follows that 
with probability at least $1-2^{-\size(p,\Phi)}$ we have 
$\kappaff(p)\leq 2^{c\,\size(p,\Phi)}$ and hence, by part~(i), 
\begin{equation*}
\cost(p,\Phi)\leq \size(\Phi) q^{\Oh(n)}
(nD)^{\Oh(n^2)} 2^{c\,\size(p,\Phi)n^2} 
\leq 2^{\Oh\big(\size(p,\Phi)^{1+\frac2D}\big)}
\end{equation*}
the last since $\size(p,\Phi)\ge N=\Omega(n^D)$.
\eproof

\subsection{Parallel computations}\label{sec:parallel}

The next result does not attempt to exhibit precise bounds. 
It only sketches a proof of weak parallel polynomial time. 

\begin{prop}\label{prop:parallel}
Algorithm {\sc Homology} parallelizes well. That is, it can be executed 
with 
$$
 \size(\Phi) (nqD\kappaff(p))^{n^{\Oh(1)}}
$$
parallel processors with a parallel time bounded by
$$
   \depth(\Phi)(n\log_2(nqD\kappaff(p)))^{\Oh(1)}
$$
where $\depth(\Phi)$ is the smallest depth of a tree with 
nodes $\vee$ and $\wedge$ evaluating $\Phi$. 

If $p$ is drawn from the uniform distribution 
on $\bbS^{N-1}$, then the parallel cost with input $(p,\Phi)$ 
is bounded by~$\size(p,\Phi)^{\Oh(1)}$ with probability at least 
$1-2^{-\size(p,\Phi)}$. 
\end{prop}

\begin{proof}
Each iteration of the repeat loop in {\sc $\kappabar$-Estimate}
can be fully parallelized. That is, the 
$(qnD\kappabar(f))^{\Oh(n)}$ evaluations done to compute $k$ 
are performed independently and then a maximum is taken with parallel 
cost $\Oh(n)\log_2(nqD\kappabar(f))$. As the loop is executed 
at most $\log_2(D\kappabar(f))+\Oh(1)$ times, the cost of 
{\sc $\kappabar$-Estimate} is well within the claimed bounds.

The $3q$ computations corresponding to the atomic sets in the 
for loop in {\sc Simplicial} are done independently. For each of them, 
we first compute the cloud $\mcX_j^{\propto_j}$ and the the 
simplicial complex $\fkA_j^{\propto_j}=\cech{\varepsilon}{\mcX_j^{\propto_j}}$. 
The computation of the cloud amounts to $(nD\kappabar(f))^{\Oh(n)}$ evaluations 
of $f$ at a point, which can be independently done. Each of them 
can be done in parallel time $\log_2 N$. Again within the claimed bound.

The sets $F_k$ of $k$-faces of $\fkA_j^{\propto_j}$ can be computed 
independently for $j=1,\ldots,q$ and $k=0,\ldots,n$. It is well-known 
that these computations parallelize well (deciding whether a $k$-tuple of 
points is a $k$-face is deciding the truth of an existential formula, 
a problem whose parallel complexity is bounded in~\cite{BaPoRo96}). 
That is, we can compute each of them in time at most 
$$
  (n\log_2(nD\kappabar(f)))^{\Oh(1)}.
$$ 
We then compute $\fkC$ in parallel time $\depth(\Phi)\Oh(n^2)\log_2(nD\kappabar(f))$. The 
techniques used to, finally, compute $H_*(\fkC)$, basic linear algebra 
and the computation of the Smith normal form, parallelize well. 

The  last part of the statement is obtained as in 
the proof of Theorem~\ref{thm:main_result}(iii).
\end{proof}

\bibliographystyle{plain}
{\small 
\bibliography{biblio}
}
\appendix
\section[On the smoothness assumption in Thom's first isotopy lemma]{On the smoothness assumption \\\hspace*{6.5cm}
in Thom's first isotopy lemma}\label{sec:proofThom}

We begin observing that we can define Whitney stratifications 
of any subset of a manifold 
in the same manner we define Whitney stratifications of 
the manifold itself. 

The following lemma will be instrumental in our proof.

\begin{lem}\label{lem:inst}
Let $\mcM$ be a smooth manifold and $\mcS$ be a locally finite partition of a locally closed 
subset $\Omega\subset\mcM$. Then:
\begin{enumerate}
\item Let $\mcS^c$ be the partition whose elements are the connected components of the elements in $\mcS$. If  
$\mcS$ is a Whitney stratification, then so is $\mcS^c$.
\item If $\mcS$ is a Whitney stratification with connected strata then it satisfies the boundary condition: 
\begin{equation}\label{eq:boCo}\tag{BC}
\mbox{for $\sigma,\varsigma\in\mcS$, if $\varsigma\cap\overline{\sigma}\neq\varnothing$, then $\varsigma\subseteq \overline{\sigma}$.}
\end{equation}
\item Let $\mcS$ satisfy the boundary condition~\eqref{eq:boCo}. 
Then $\mcS$ is a Whitney stratification if and only if for all 
$\sigma\in\mcS$, 
\[
  \mcS_{|\overline{\sigma}}:=\{\varsigma\in\mcS\mid\varsigma\subseteq\overline{\sigma}
  \}
\]
is a Whitney stratification of $\overline{\sigma}$.
\item Let $\mcS$ be a Whitney stratification,  
$\{\sigma_i\}_{i\in I}\subseteq\mcS$ a family of strata 
of the same dimension and $\mcS'$ the partition obtained from $\mcS$ by replacing the $\sigma_i$ by its union. Then $\mcS'$ is 
a Whitney stratification.
\end{enumerate}
\end{lem}

\begin{proof}
Parts~1 and~2 are~\cite[Ch.~II, Theorem~5.6 and Corollary~5.7]{gibson}, respectively. 
For part~3, the fact that $\mcS$ satisfies~\eqref{eq:boCo} 
implies that $\mcS_{|\overline{\sigma}}$ is a partition 
of $\overline{\sigma}$ whose elements are elements in 
$\mcS$. As a subset of a Whitney stratification is 
a Whitney stratification  itself we have shown the `only if' 
part. We next show the converse. For every $\sigma\in\mcS$, 
the fact that $\mcS_{|\overline{\sigma}}$ is a Whitney 
stratification implies that 
$\sigma\in\mcS_{|\overline{\sigma}}$ is a locally closed 
smooth submanifold. Next note that Whitney's condition~b 
needs to be checked only for pairs 
$(\sigma,\varsigma)\in\mcS^2$ such that $\varsigma\cap\overline{\sigma}\neq\varnothing$. 
But the fact that $\mcS$ satisfies~\eqref{eq:boCo} 
implies that, for any such pair, 
$\sigma,\varsigma\in\mcS_{|\overline{\sigma}}$ and 
therefore, it satisfies condition~b because, by 
hypothesis, $\mcS_{|\overline{\sigma}}$ is a Whitney 
stratification.

We finally prove~4. By the local character of Definition~\ref{defiwhitney} of
Whitney stratification, it is enough to check the conditions 
in this definition in some open neighborhood $U_x$ around 
each point $x\in\Omega$. Since $\mcS$ is locally finite, we 
can pick each $U_x$ such that $\mcS_{|U_x}$ is finite. 
Hence, without loss of generality, we can assume that 
$I$ is finite. 

For all $i\ne j\in I$ we have 
$\sigma_i\cap\overline{\sigma_j}=\varnothing$. Otherwise, 
by~\cite[Ch.~I, (1.1)]{gibson}, we would have $\dim\sigma_i<
\dim\sigma_j$, contradicting our hypothesis. Hence, for 
all $i\in I$, there is an open set $U_i$ such that 
$\sigma_i\subseteq U_i$ and, for all $j\ne i$, 
$U_i\cap\overline{\sigma_j}=\varnothing$. It follows that 
$\cup \sigma_i$ is a locally closed smooth manifold. 
The verification of the conditions in Definition~\ref{defiwhitney} for $\mcS'$ is now 
straightforward.
\end{proof}

To prove Theorem~\ref{thomfirstlemmaB} we will rely on 
the following 
version of Thom's first isotopy lemma which is the one 
in~\cite[Ch.~II, Theorem~5.2]{gibson}.

\begin{theo}
\label{thomfirstlemmaSMOOTH}
Let $\mcM$ be a smooth manifold and $\Omega\subseteq \mcM$ 
a locally closed subset with a Whitney stratification~$\mcS$ 
and let $\alpha:\mcM\rightarrow \bbR^k$ be a smooth proper map such that: 
\begin{itemize}
\item 
for each stratum $\sigma\in\mcS$, $\alpha_{|\sigma}:\sigma\rightarrow \bbR^k$ is surjective,
\item 
for each stratum $\sigma\in\mcS$, $\alpha_{|\sigma}:\sigma\rightarrow \bbR^k$ is a 
smooth submersion.
\end{itemize}
Then $\alpha_{|X}$ is a trivial fiber bundle. \eproof
\end{theo}

To deduce Theorem~\ref{thomfirstlemmaB} from this result we will 
employ graphs of maps. This will allow us to transform our 
not necessarily 
smooth map into a smooth one, as it will be simply a projection.

Let $A$ and $B$ be smooth manifolds. Recall that the 
\emph{graph} of a function 
$\varphi:A\rightarrow B$ is the set
\[
  \Gamma_\varphi:=\{(a,b)\in A\times B\mid \varphi(a)=b\}.
\]
Associated with the graph we have the functions 
$i_\varphi:A\rightarrow \Gamma_\varphi$, given by $a\mapsto (a,\varphi(a))$, and 
$\pi:A\times B\rightarrow B$, given by $(a,b)\mapsto b$. 
Clearly, $\varphi=\pi\circ i_\varphi$. Also, it is easy to see, if $\varphi$ is a continuous map, 
then $\Gamma_\varphi$ is a closed subset of $A\times B$ and $i_\varphi$ is a homeomorphism between 
$A$ and $\Gamma_\varphi$. Finally, if $\varphi$ is a smooth map, then $\Gamma_\varphi$ is a closed 
smooth submanifold of $A\times B$ and $i_\varphi$ is a diffeomorphism between $A$ 
and $\Gamma_\varphi$. Given a subset $X\subseteq A$, 
we will consider 
\[
 \Gamma_\varphi(X):=\Gamma_{\varphi_{|X}}=
 \{(a,b)\in\Gamma_\varphi\mid a\in X\}=
 \{(a,\varphi(a))\mid a\in X\}.
\]
It is again clear that if $\varphi$ is continuous and $X$ is a 
locally closed subset of $A$, then $\Gamma_\varphi(X)$ is 
a locally closed subset of 
$A\times B$. Moreover, if $\varphi$ is smooth and $X$ is a locally closed smooth submanifold 
of $A$, then $\Gamma_\varphi(X)$ is a locally closed smooth submanifold of $A\times B$.

\begin{proof}[Proof of Theorem~\ref{thomfirstlemmaB}]
Consider the graph $\Gamma_{\alpha}$ of $\alpha$. 
Although not necessarily a manifold (as $\alpha$ may be non-smooth), it is a locally closed 
subset of $\mcM\times\bbR^k$. 
Next consider the partition of $\Gamma_\alpha$ given by
\[
\Gamma_\alpha(\mcS):=
\{\Gamma_\alpha(\sigma)\mid \sigma \in\mcS\}
\]
and its associated partition $\Gamma^c_\alpha(\mcS)$ 
as defined in Lemma~\ref{lem:inst}(1). 

We claim that $\Gamma^c_\alpha(\mcS)$ is a Whitney 
stratification of 
$\Gamma_\alpha$. 

To prove the claim we first observe that 
$\Gamma^c_\alpha(\mcS)=\Gamma_\alpha(\mcS^c)$. As, by 
Lemma~\ref{lem:inst}(1), 
$\mcS^c$ is a Whitney stratification and, by construction, 
has connected strata, Lemma~\ref{lem:inst}(2) shows that 
it satisfies the boundary condition~\eqref{eq:boCo}. 
It follows that $\Gamma^c_\alpha(\mcS)$ 
satisfies~\eqref{eq:boCo} as well.

Let $\sigma \in \mcS^c$ and $\sigma'\in\mcS$ 
such that $\sigma\subseteq\sigma'$. By the first hypothesis 
in our statement, there is an open neighborhood $U$ 
of $\overline{\sigma'}\supseteq\overline{\sigma}$ 
and a smooth map 
$\varphi:U\rightarrow \bbR^k$ such that $\alpha_{|\sigma'}=\varphi$. Clearly, 
$\alpha_{|\sigma}=\varphi$ as well. 
This implies that 
$\Gamma_\alpha(\mcS^c_{|\overline{\sigma}})
=\Gamma_\varphi(\mcS^c_{|\overline{\sigma}})$ 
and $\Gamma_\alpha(\overline{\sigma})
=\Gamma_\varphi(\overline{\sigma})$. 
Since $\varphi$ is smooth, 
$\Gamma_\varphi$ is a locally closed smooth submanifold and $i_\varphi:U\rightarrow \Gamma_\varphi$ is a 
diffeomorphism mapping the Whitney stratification $\mcS^c_{|\overline{\sigma}}$ to
$\Gamma^c_\alpha(\mcS_{|\overline{\sigma}})$ and 
the closed set $\overline{\sigma}$
to $\Gamma_\alpha(\overline{\sigma})$. Hence, 
by~\cite[Ch.~I, (1.4)]{gibson},
$\Gamma^c_\alpha(\mcS_{|\overline{\sigma}})$ is a 
Whitney stratification of $\Gamma_\alpha(\overline{\sigma})$. 

As this happens for all strata 
$\Gamma^c_\alpha(\overline{\sigma})$ of the partition 
$\Gamma^c_\alpha(\mcS)$ we may apply Lemma~\ref{lem:inst}(3) 
to deduce that $\Gamma^c_\alpha(\mcS)$ is a Whitney 
stratification. We finally apply Lemma~\ref{lem:inst}(4) 
(several times for each dimension) 
to deduce that $\Gamma_\alpha(\mcS)$ itself is a Whitney 
stratification. This proves the claim. 

Since $\Gamma_\alpha(\mcS)$ is a Whitney stratification 
of $\Gamma_\alpha$ the map 
$i_{\alpha}$ restricts to a diffeomorphism between $\sigma$ 
and $\Gamma_\alpha(\sigma)$, for all $\sigma\in\mcS$. 
In addition, as $\alpha=\pi\circ i_\alpha$, we have 
$\alpha_{|\sigma}=\pi_{|\Gamma_\alpha(\sigma)}\circ (i_\alpha)_{|\sigma}$ and, hence,  
as $(i_\alpha)_{|\sigma}$ is a diffeomorphism, $\pi_{|\Gamma_\alpha(\sigma)}$ is surjective 
if and only if $\alpha_{|\sigma}$ is so, and $\pi_{|\Gamma_\alpha(\sigma)}$ is a smooth submersion 
if and only if so is $\alpha_{|\sigma}$. In summary,  
the last two hypotheses of our statement  
imply the hypothesis of Theorem~\ref{thomfirstlemmaSMOOTH}, and consequently, 
that $\pi_{|\Gamma_\alpha}$ is a trivial bundle.

We can now conclude because a trivialization
$h:\Gamma_\alpha\rightarrow F\times \bbR^k$ 
of $\pi_{|\Gamma_\alpha}$ induces the trivialization 
$h\circ i_\alpha$ of $\alpha:\mcM\rightarrow \bbR^k$. 
\end{proof}

\end{document}